%% file: uplink_energy_paper_arxiv.tex
\documentclass[draftcls,onecolumn]{IEEEtran}
\pdfoutput=1
\ifCLASSOPTIONcompsoc
  \usepackage[nocompress]{cite}
\else
  \usepackage{cite}
\fi
\ifCLASSINFOpdf
\else
\fi

\usepackage{setspace}
 \usepackage{xargs}

\usepackage{amsopn}
\usepackage{amsxtra,amssymb,amsmath}
\usepackage[latin1]{inputenc}
\usepackage[english]{babel}
\usepackage{amsmath}
\usepackage{graphicx}
\usepackage{caption}
\usepackage{epstopdf}
\usepackage{framed}
\usepackage{amsfonts}
\usepackage{relsize}
\usepackage{amsthm}
\usepackage{url}
\usepackage{amssymb}
\usepackage{bbold}
 \usepackage{wasysym}
 \usepackage{color}
 \usepackage[document]{ragged2e}
 
\theoremstyle{plain}
\newtheorem{lemma}{Lemma}
\newtheorem{theorem}{Theorem}

\newtheorem{corollary}{Corollary}
\newtheorem{proposition}{Proposition}

\theoremstyle{definition}

\def\config{\mathbf{x}}
\def\objective{{\Pi}} 
\def\constraint{{\cal D}} 
\def\constraintmax{D_{\mbox{{\tiny max}}}}
\def\multiplier{\mu}
\def\augmented{{\cal F}}
\def\feasable{\tilde{\Omega}}
\def\objective{{\Pi}} 
\def\constraint{D_c} 
\def\constraintmax{D_{\mbox{{\tiny max}}}}
\def\potential{\Phi}
\def\multiplier{\mu}
\def\augmented{{\cal F}}
\def\feasable{\tilde{\Omega}}

\def\constraintMinU{U^*} 

\def\slack{t}

\def\Real{\mathbb{R}}
\def\st{\mbox{ s.t. }} 
\def\ie{\mbox{\em i.e., }} 
\def\wrt{\mbox{\em wrt. }} 

\def\Real{\mathbb{R}}
\def\st{\mbox{ s.t. }} 
\def\ie{\mbox{\em i.e., }} 

\def\ul{\underline}
\def\ra{\rightarrow}
\def\potential{\Phi}
\input{./battail.sty}

\def\charact{\bbbone}
\def\dsp{\displaystyle}
\def\cond{\ |\ }

\begin{document}

\title{Uplink Energy-Delay Trade-off under Optimized Relay Placement in Cellular Networks}

\author{Mattia Minelli,
        Maode Ma,
        Marceau Coupechoux,
        Jean-Marc Kelif,
        Marc Sigelle,
        and Philippe Godlewski 
\IEEEcompsocitemizethanks{
	\IEEEcompsocthanksitem M. Minelli is with the School of Electrical and Electronic Engineering, Nanyang Technological University, 50 Nanyang Avenue Singapore 639798 and with T\'el\'ecom ParisTech and CNRS LTCI, 46, rue Barrault, Paris, France.\protect\\
E-mail: mattia.minelli@gmail.com
	\IEEEcompsocthanksitem M. Maode is with the School of Electrical and Electronic Engineering, Nanyang Technological University \protect \\
	E-mail: EMDMa@ntu.edu.sg
	\IEEEcompsocthanksitem M. Coupechoux, M. Sigelle and P. Godlewski are with T\'el\'ecom ParisTech and CNRS LTCI. \protect \\
	E-mails: \{first name.name\}@telecom-paristech.fr
	\IEEEcompsocthanksitem J.-M. Kelif is with Orange Labs, Issy-Les-Moulineaux, France. \protect \\
	Email: jeanmarc.kelif@orange.com
	}
}

%





\maketitle

\begin{abstract}
Relay nodes-enhanced architectures are deemed a viable solution to enhance coverage and capacity of nowadays cellular networks. Besides a number of desirable features, these architectures reduce the average distance between users and network nodes, thus allowing for battery savings for users transmitting on the uplink. In this paper, we investigate the extent of these savings, by optimizing relay nodes deployment in terms of uplink energy consumption per transmitted bit, while taking into account a minimum uplink average user delay that has to be guaranteed. A novel performance evaluation framework for uplink relay networks is first proposed to study this energy-delay trade-off. A simulated annealing is then run to find an optimized relay placement solution under a delay constraint; exterior penalty functions are used in order to deal with a difficult energy landscape, in particular when the constraint is tight. Finally, results show that relay nodes deployment consistently improve users uplink energy efficiency, under a wide range of traffic conditions and that relays are particularly efficient in non-uniform traffic scenarios.
\end{abstract}
\vspace{1cm}
\justify

\section{Introduction}

\input{./introduction.tex}

\section{System Model} \label{system_model}

\input{./system_model.tex}

\section{Performance Evaluation} \label{perf_eval}

\input{./perf_eval.tex}

\section{Optimization} \label{optimization}

\input{./optimization.tex}

\section{Results}\label{results}

\input{./results.tex}


\section{Conclusion}\label{conclusions}

\input{./conclusions.tex}

\appendices

\input{appendix.tex}

\begin{IEEEbiography}[{\includegraphics[width=1in,height=2.25in,clip,keepaspectratio]{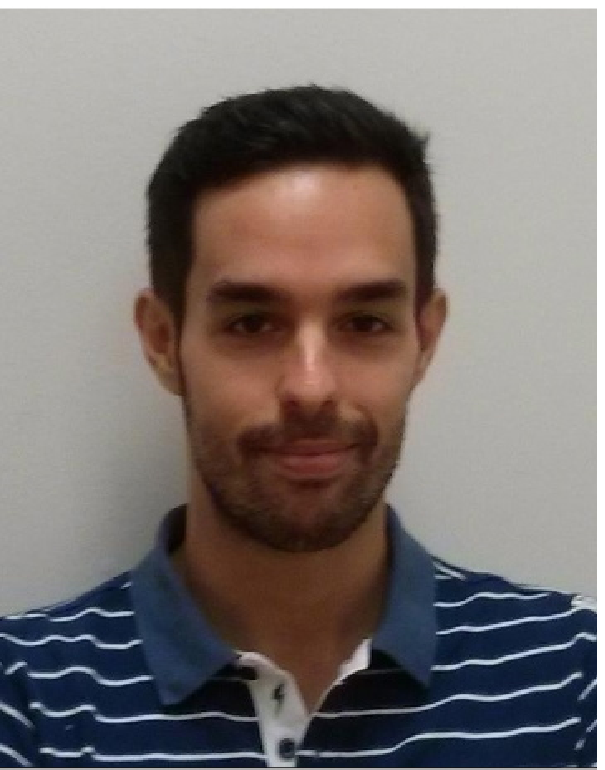}}]{Mattia Minelli} received his Master degree in telecommunications
engineering from Politecnico di Milano in Italy, in 2009. He is currently
participating in a Joint PhD Program between the School of
Electrical and Electronic Engineering at Nanyang Technological University
in Singapore and the Computer and Network Science Department
of Telecom ParisTech in Paris. His research interests include
wireless networks, propagation and radio resource management issues,
and relaying in OFDMA networks.
\end{IEEEbiography}

\begin{IEEEbiography}[{\includegraphics[width=1in,height=2.25in,clip,keepaspectratio]{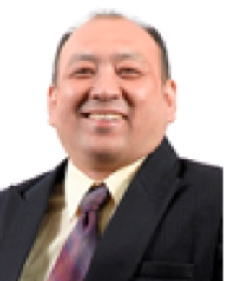}}]{Ma Maode} received his PhD degree in computer science from Hong Kong University of Science and Technology in 1999. Now, Dr. Ma is an Associate Professor in the School of Electrical and Electronic Engineering at Nanyang Technological University in Singapore. He has extensive research interests including wireless networking and network security. Dr. Ma has more than 250 international academic publications including over 100 journal papers and more than 130 conference papers. He currently serves as the Editor-in-Chief of {\it International Journal of Electronic Transport}. He also serves as an Associate Editor for other five international academic journals. Dr. Ma is an IET Fellow and a senior member of {\it IEEE Communication Society} and {\it IEEE Education Society}. He is the vice Chair of the {\it IEEE Education Society}, Singapore Chapter. He is also an {\it IEEE Communication Society Distinguished Lecturer}.
\end{IEEEbiography}

\vspace{-1.4cm}
\begin{IEEEbiography}[{\includegraphics[width=1in,height=2.25in,clip,keepaspectratio]{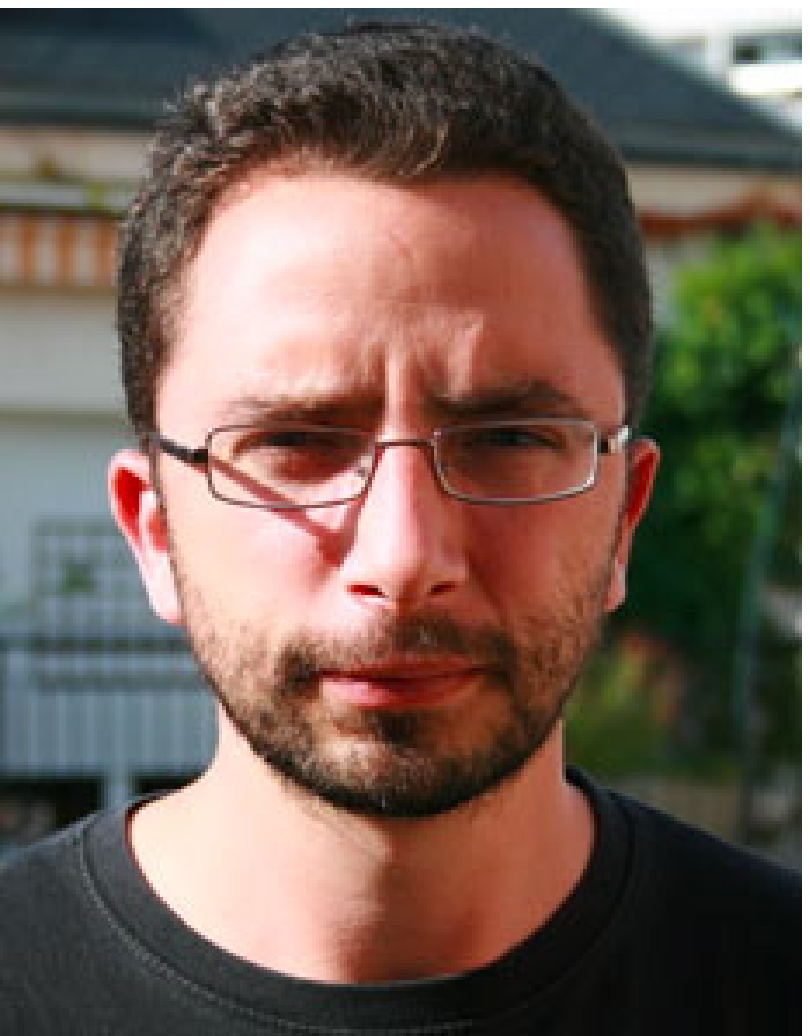}}]{Marceau Coupechoux} has been working as an Associate Professor at Telecom ParisTech since 2005. He obtained his Masters' degree from Telecom ParisTech in 1999 and from University of Stuttgart, Germany, in 2000, and his Ph.D. from Institut Eurecom, Sophia-Antipolis, France, in 2004. From 2000 to 2005, he was with Alcatel-Lucent. He was a Visiting Scientist at the Indian Institute of Science, Bangalore, India, during 2011--2012. Currently, at the Computer and Network Science department of Telecom ParisTech, he is working on cellular networks, wireless networks, ad hoc networks, cognitive networks, focusing mainly on radio resource management and performance evaluation.
\end{IEEEbiography}

\vspace{-1.1cm}
\begin{IEEEbiography}[{\includegraphics[width=1in,height=2.25in,clip,keepaspectratio]{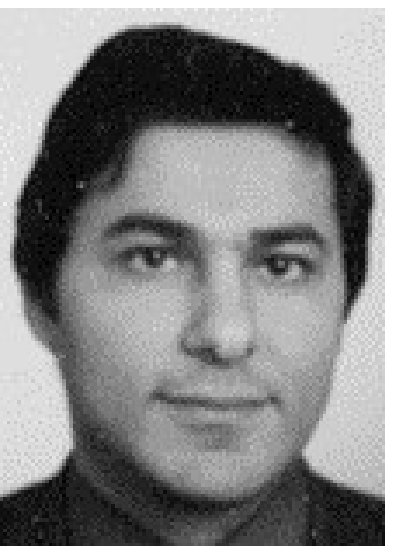}}]{Jean-Marc Kelif} received the Licence of Fundamental Physics Sciences
from the Louis Pasteur University of Strasbourg (France), and
the Engineer degree in Materials and Solid State Physics Sciences from
the University of Paris XIII in 1984. He worked
as engineer in materials, and specialized in telecommunications. Since
1993, he has been with Orange Labs in Issy-les-Moulineaux, France. His current research interests include the performance evaluation, modeling, dimensioning
and optimization of wireless telecommunication networks.
\end{IEEEbiography}

\vspace{-1.1cm}
\begin{IEEEbiography}[{\includegraphics[width=1in,height=2.25in,clip,keepaspectratio]{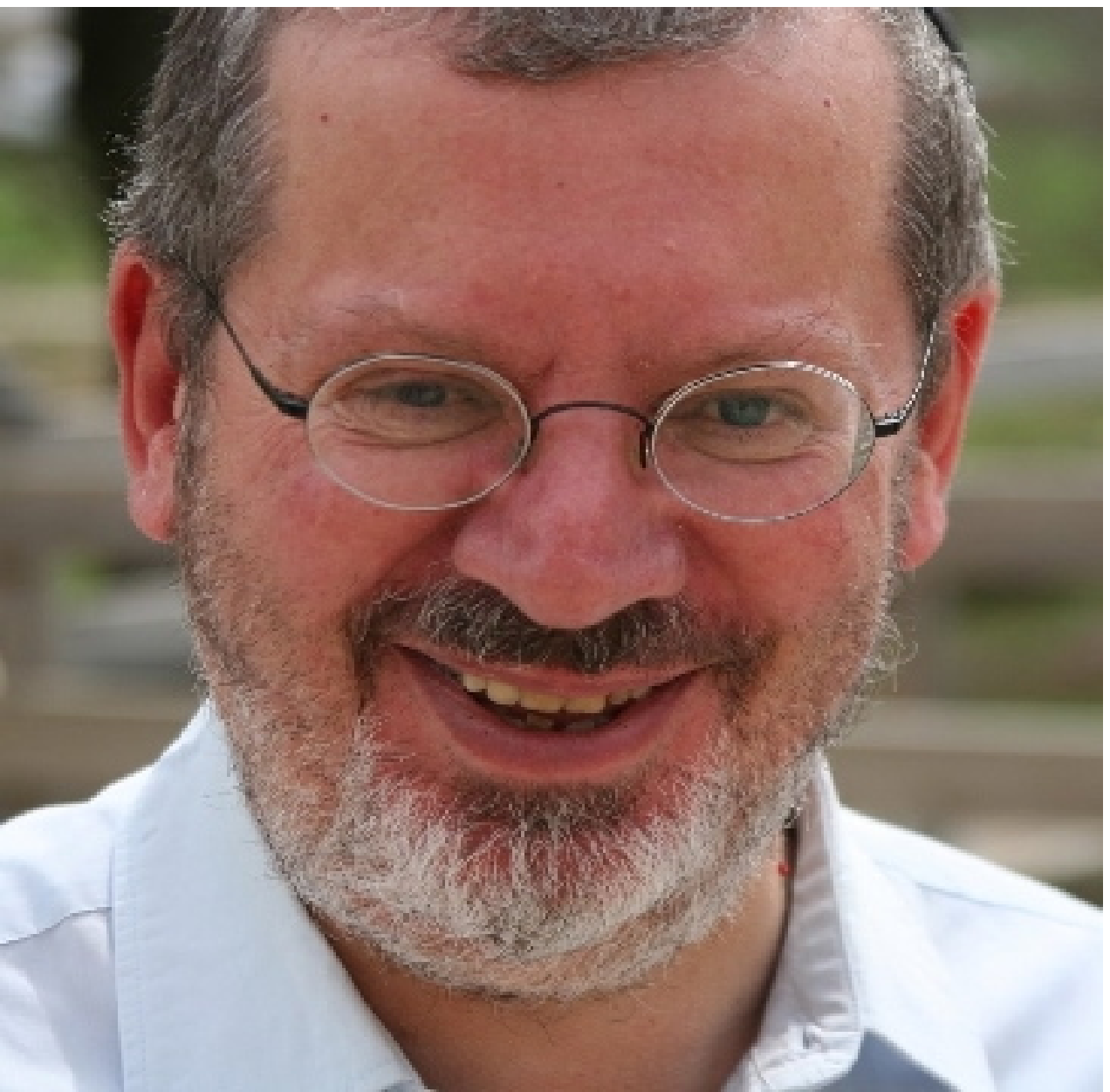}}]{Marc Sigelle} graduated as an Engineer from Ecole Polytechnique in 1975 and from Telecom ParisTech in 1977. He obtained a PhD from Telecom ParisTech in 1993. He worked first at Centre National d'Etudes des Télécommunications in physics and computer algorithms. Since 1989 he has been working at Telecom ParisTech in image processing, his main topics of interests being the restoration and segmentation of images using stochastic/statistical models and methods (Markov Random Fields, Bayesian networks, Graph Cuts in relationships with Statistical Physics). 
In the past few years he has been involved in the transfer of this knowledge to the modeling and optimization (namely, stochastic) of wireless networks. 
\end{IEEEbiography}

\vspace{-1.1cm}
\begin{IEEEbiography}[{\includegraphics[width=1in,height=2.25in,clip,keepaspectratio]{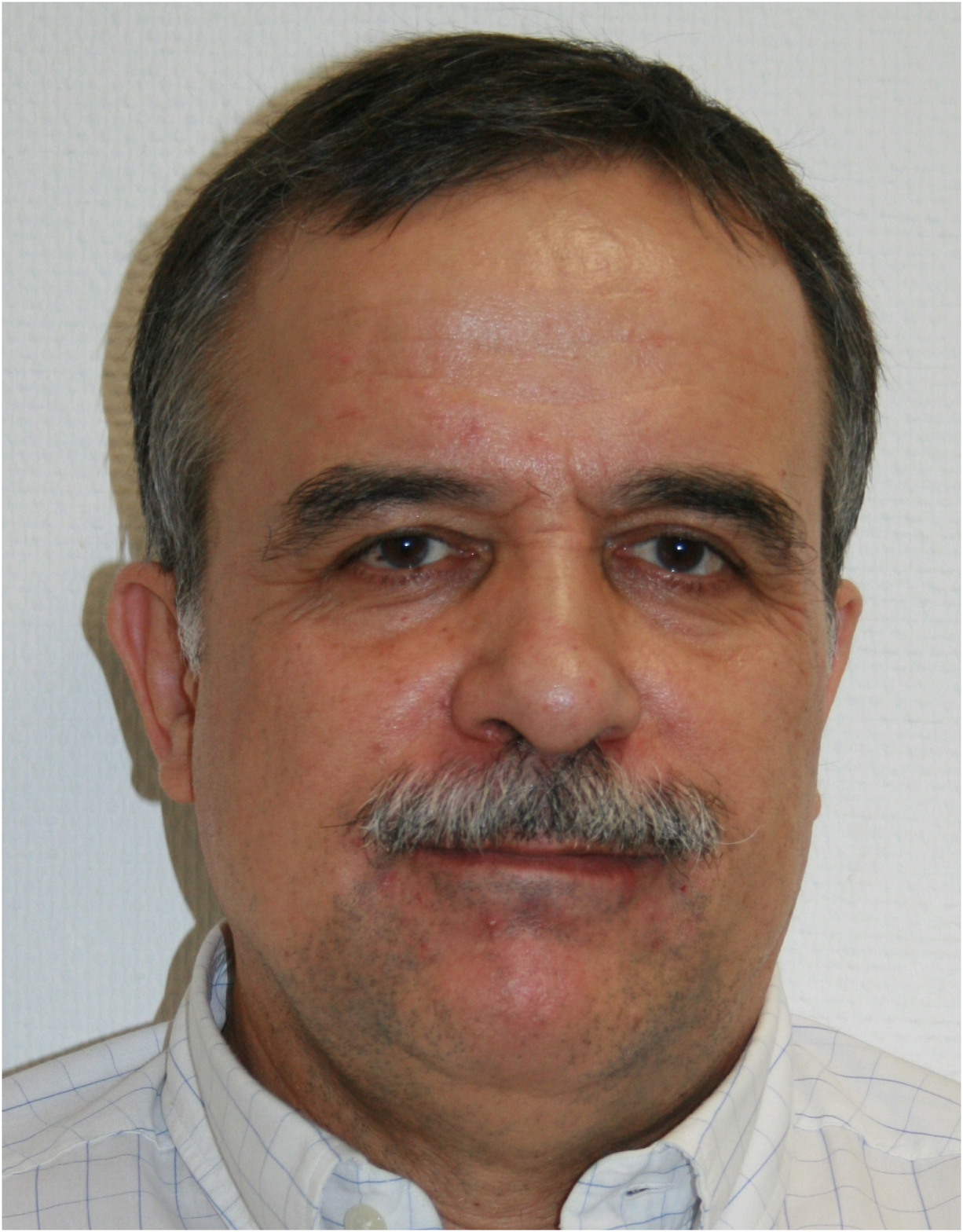}}]{Philippe Godlewski} received the Engineer's degree from Telecom Paris (formerly ENST) in 1973 and the PhD in 1976 from University Paris 6. He is Professor at Telecom ParisTech in the Computer and Network Science department. His fields of interest include Cellular Networks, Air Interface Protocols, Multiple Access Techniques, Error Correcting Codes, Communication and Information Theory.
\end{IEEEbiography}

\end{document}

%% file: introduction.tex
Heterogeneous network deployments in modern cellular networks are regarded as a promising solution to meet the ever-increasing demand of wireless data and voice traffic. They consist in installing a number of low-power nodes, possibly of different types (e.g., femtocells, Relay Nodes (RNs), etc.), inside the coverage area of macro base stations (also called eNodes-B (eNBs) in this paper), resulting in a more {\it dense} network architecture.
Indeed, a higher density of nodes entails a number of benefits, e.g., coverage and capacity are boosted \cite{relays_pabst} 
and power efficiency improves \cite{hou}, due to the reduced distance between User Equipments (UE) and serving nodes \cite{hasan}. This generates operational expenditures savings for operators, and lowers the environmental impact of their infrastructure.    
Results in \cite{hou, saleh12} et. al. show that the introduction of RNs and picocells can reduce downlink power consumption. However, this applies to the uplink as well (see e.g. \cite{Zhao,res_opt}), and the decrease in battery energy consumption for users can be consistent: \cite{uplink_consumption,uplink_consumption_CA} indicate that uplink transmit power has a strong impact on users overall energy demand, especially for high transmit powers.   

Hence, heterogeneous networks can be considered as an effective means to extend mobile users batteries duration, and several research works are dedicated to this topic. Authors of \cite{Zhao} and \cite{Chen} propose a game theory-based framework for the maximization of femtocells uplink energy efficiency, by means of users transmit powers tuning \cite{Zhao} or radio resources assignment and power control parameters optimization \cite{Chen}. In \cite{Zhang}, closed-form expressions of the Signal to Interference Ratio (SIR) and of the outage probability are derived and used to measure the impact of femtocells coverage and users density on energy efficiency. Both \cite{Zhao} and \cite{Chen} apply Quality of Service (QoS)-related constraints to the optimization problem, so as to avoid solutions with high energy efficiency but poor performance. Other works in the area of cellular networks, e.g. \cite{son}, address this issue by jointly optimizing power and user experienced QoS. 

The close interaction between energy consumption and performance \cite{son} makes the problem of energy efficiency maximization in relay-enhanced networks different from that addressed in \cite{Zhao,Chen} et al. for the femtocells deployment case; this is because RNs communicate to their {\it donor} node through a wireless backhaul link \cite{3gpprelayperf}, which represents a well-known performance bottleneck. On the contrary, small cells and femto cells benefit from a wired backhaul, where the delay issue is less crucial.      
Surprisingly, uplink RN networks energy consumption has received limited attention in the literature, to the best of our knowledge. 
Reference \cite{bulakci13} proposes an optimization of the uplink power control in relay-based networks but ignores energy consumption. Authors of \cite{res_opt, Liang11} treat the maximization of uplink energy efficiency via either optimal assignment of subcarriers, users powers and bit allocations \cite{res_opt}, or optimal radio resources allocation \cite{Liang11}. More than one order of magnitude can be achieved in user power consumption \cite{Liang11}.

A drawback of \cite{res_opt} is that the decisive impact of co-channel interference is neglected. Also, the paper assumes that users are fixed in number and position, and always have data to transmit. Hence, the influence of the traffic intensity is not considered, while a number of studies (see e.g. \cite{nagaraj,son}) show the importance of the traffic load for energy efficiency evaluation. 
Finally, \cite{res_opt,Liang11} lack of a thorough theoretical framework for the analysis of relay-enhanced cellular networks uplink energy efficiency, 
and they are based on the sole minimization of users energy consumption, without investigating the necessary tradeoff between this consumption and uplink performance. Hence, the need arises for a finer model to study this tradeoff. 



The contributions of this paper are the following.
\begin{itemize}
\item We study the tradeoff between energy consumption and users experienced delay in uplink relays cellular networks. This is the first study of this type, to the best of our knowledge. We show that in many cases, relays can help both saving energy and reducing delays despite the constraint imposed by the wireless backhaul. We also highlight the interest of using relays in scenarios, where traffic is not uniform. 

In order to find upper bounds on the achievable gains in terms of energy saving and/or average delay, we formulate a constrained optimization problem. Its objective is to place the relays and tune the network parameters so as to minimize the energy consumption per transmitted bit under delay constraints. To solve this problem, we rely on a Simulated Annealing (SA) algorithm enhanced with exterior search and penalty functions. We discuss the effectiveness of our approach for addressing our problem.

\item We propose a framework for the analysis of energy efficiency in uplink relay networks. In order to study the delay constraint, we consider the dynamic nature of cellular traffic, by means of a model of UE arrivals and departures. This results in a hierarchical flow level analysis. The loads of eNBs and RNs are accounted in the estimation of interference and transmission delays. Our propagation model includes shadowing and fast fading, and UEs power control is considered. Overall, the proposed model is more comprehensive, compared to the existing literature on uplink energy efficiency in heterogeneous networks.    

\item We adopt a non-trivial scheduling scheme, which represents an approximation of the Proportional Fair (PF) scheduler, and derive the probability density function of the Signal to Interference plus Noise Ratio (SINR) of a scheduled UE. The choice of the scheduling algorithm has indeed a decisive impact on the delay performance. This is a novel contribution, to the best of our knowledge.  
\end{itemize}

The paper is organized as follows. Section~\ref{system_model} introduces our system model, while Section~\ref{perf_eval} describes the framework used for evaluating energy consumption per transmitted bit and delay. Section~\ref{optimization} is devoted to the optimization problem and our proposed algorithm, while Section~\ref{results} gives our numerical results. Finally, Section~\ref{conclusions} concludes our work. 

%% file: system_model.tex
We work on the uplink of a cellular network, where transmissions are performed on synchronized frames, and each network {\it station} (i.e. eNB or RN) is associated an uplink frame. Frames are partitioned into Radio Blocks (RBs) of the same bandwidth and time duration, which are labeled with an index $\tau$. Each RB $\tau$ in the frame of station $k$ can be either granted to one of the UEs served by $k$, or be left unused (if there is no UE to be served).  

The network is composed of a set $\mathcal{K}$ of $K$ stations, and divided into cells; let $\mathcal{K}_B$ be the set of eNBs, and $\mathcal{K}_R$ the set of RNs, so that $\mathcal{K}=\mathcal{K}_B \cup \mathcal{K}_R$.
We focus on one given cell $c$ containing one {\it donor} eNB \cite{36.814} and $n_{RN}$ RNs, which are connected to the eNB by means of a wireless {\it backhaul link} (Fig.~\ref{fig:uplink_model}). The set of stations of cell $c$ is denoted $\mathcal{K}_c$. The surface $\mathcal{A}_c$ of cell $c$ is defined as the region where UEs are served by one of the stations of $c$, and its area is denoted $A_c$. Similarly, surface $\mathcal{A}_k$, of area $A_k$, is the region where UEs are served by station $k$. 

UEs select their serving station based on the highest product of received downlink pilot power times a station specific {\it bias}. They are then served over the {\it access link}. The technique of biasing the user association is referred to in literature as {\it cell range expansion} \cite{Damnjanovic}, and it is deemed beneficial for network performance as it allows for load balancing. We assume that eNBs have a bias of 1 and RNs have a common bias $\mathcal{B}$.
\begin{figure}
\centering
\begin{minipage}{.47\textwidth}
\centering
\includegraphics[width=\linewidth]{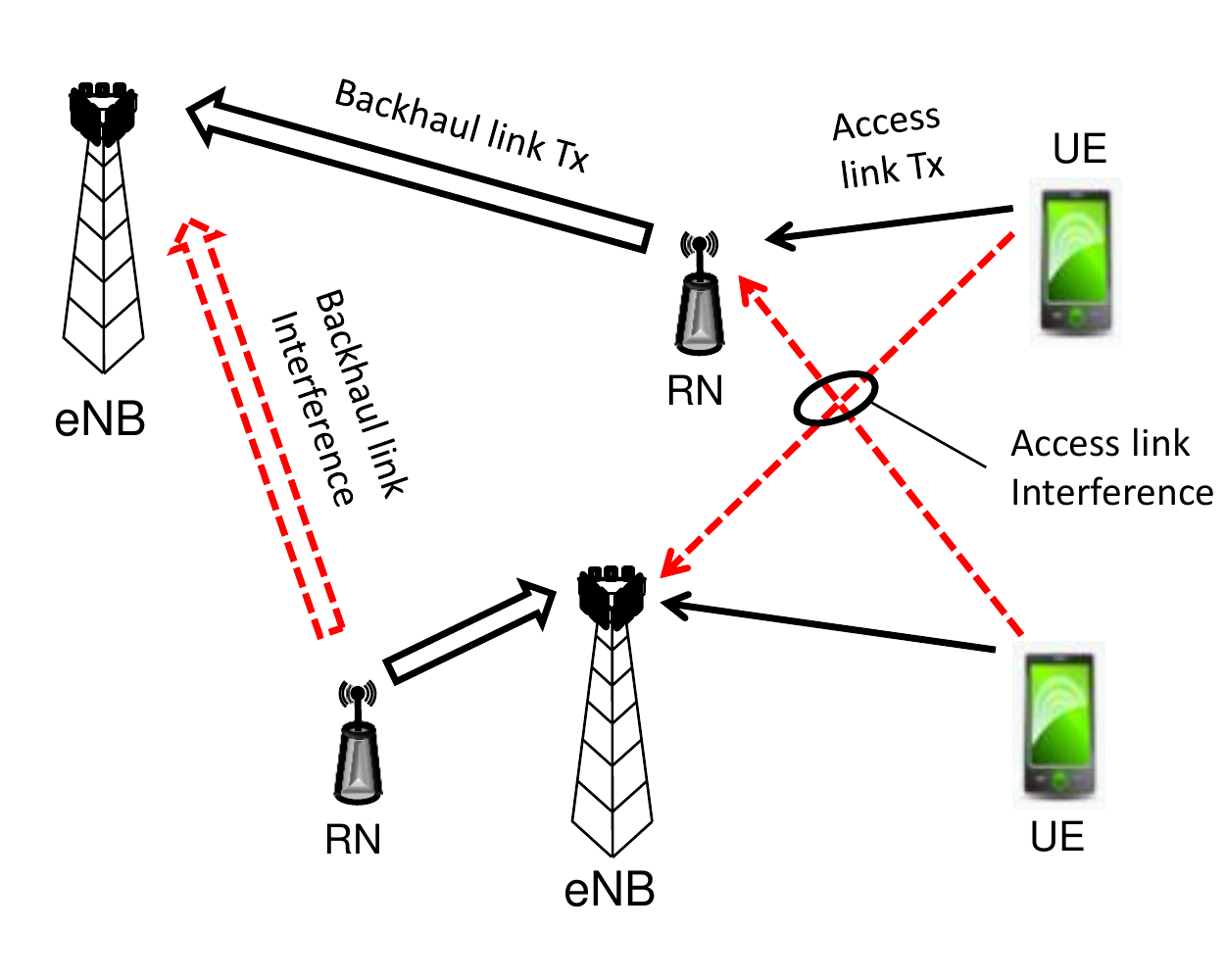} 
\caption{Uplink model.}
\label{fig:uplink_model} 
\end{minipage}\hfill
\begin{minipage}{.47\textwidth}
\centering
\includegraphics[width=\linewidth]{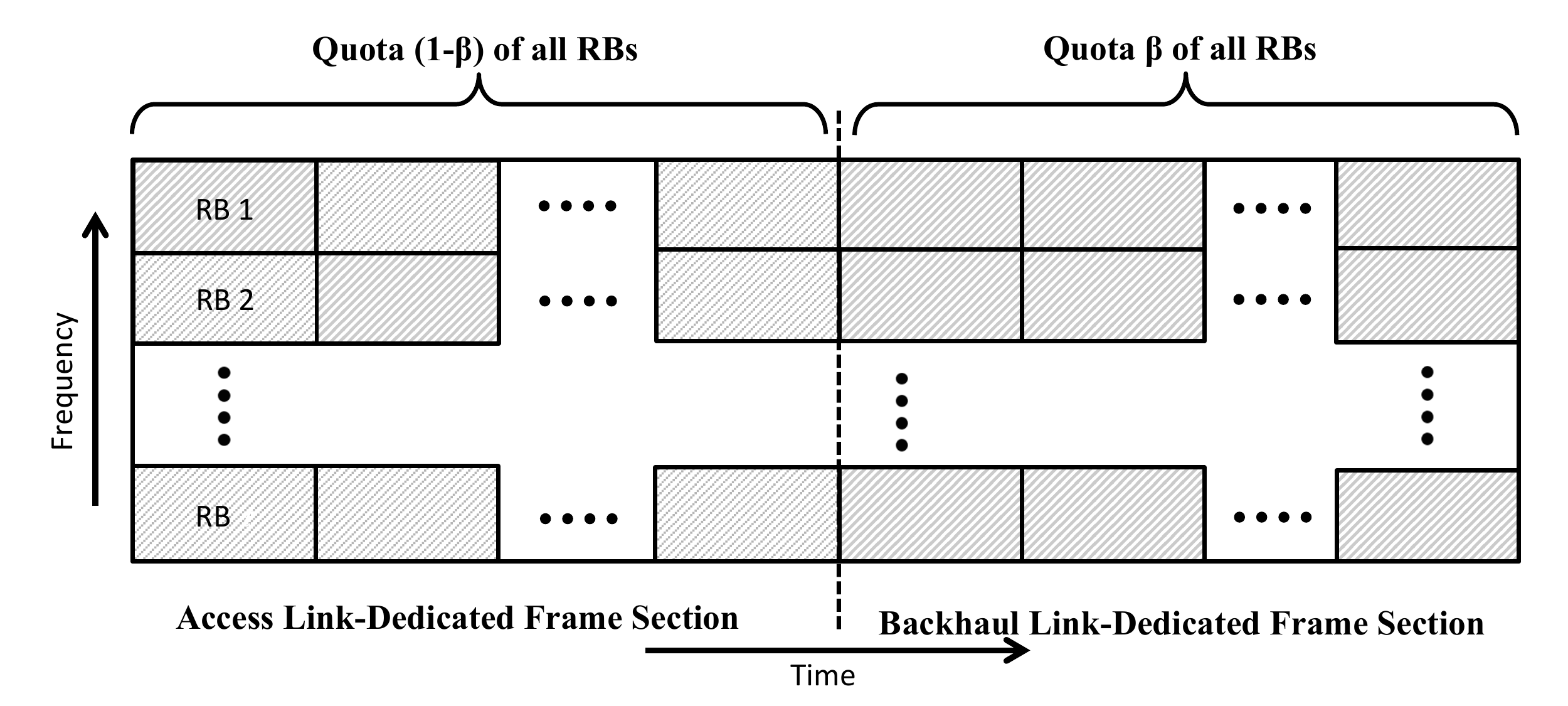}
 \caption{Frame structure.}
\label{fig:frame_struct} 
\end{minipage}
\end{figure}

\subsection{Frame Structure}
Access link and backhaul link transmissions are dedicated two orthogonal sections of the frame (Fig.~\ref{fig:frame_struct}). RN transmissions on the backhaul link, which uses a quota $\beta$ of the RBs, do not interfere with UEs transmissions on the access link.
This choice is widely adopted in the literature (see e.g. \cite{bulakci13, Liang11}), as it avoids interference between RNs and UEs on the uplink. 
For the same reason, we assume that $\beta$ is the same for all network stations. The value of $\beta$ is supposed to be set by the operator and based on considerations on the overall network performance, while our analysis is focused on the performance related to one given network cell, {\it given} $\beta$. 

\subsection{Traffic Model}
We assume that UEs arrive in the network according to a spatial Poisson point process of intensity $\lambda(s)\hspace{5 pt}[arrivals/s/m^2]$, transmit a flow of average size $\xi \hspace{5 pt} [bits]$ to their serving station, and leave the network. 
Flows transmitted from UEs to a RN on the access link are then forwarded by the RN to its donor eNB on the backhaul link. 

The traffic density $\omega(s)$ at a given location $s$ is denoted by $\omega(s)=\lambda(s)\xi \hspace{5 pt}[bits/s/m^2]$, while the average traffic density $\bar{\omega}$ in the network can be computed as  $\bar{\omega}=(1/A_{nw})\int_{\mathcal{A}_{nw}} \lambda(s)\xi {ds}\hspace{5 pt}[bits/s/m^2]$, where $\mathcal{A}_{nw}$ is the overall network surface, of area $A_{nw}$\footnote{We assume $A_{nw}$ large enough so that interference in cell $c$ is accuratly computed.}. We define $\phi(s)\triangleq\omega(s)/\bar{\omega},\forall s \in \mathcal{A}_{nw}$, to account for the local variation of the traffic density with respect to $\bar{\omega}$. 
The ratio $\phi(s)/A_{nw}$ can be seen as the spatial Probability Density Function (PDF) of UEs arrivals. We also define $\Phi_k\triangleq \int_{\mathcal{A}_k}\phi(s){ds}$ for the sake of further developments.

The typical number of bits carried by a RB is assumed to be {\it much smaller} than $\xi$. Hence, access link buffer queue of all stations $k\in \mathcal{K}$ can be modeled as an M/G/1/PS queue (Fig.~\ref{fig:queues}), and its load is denoted with $\rho_k$. It is the sum of load contributions $\varrho_k(s)$ over $\mathcal{A}_k$, so that 
$\rho_k= \int_{\mathcal{A}_k} \varrho_k(s){ds}$. Only {\it stable} scenarios, i.e., $\rho_k<1, \forall k \in \mathcal{K}$ are considered in this work. We define the vectors $\boldsymbol{\rho}=[\rho_1\; \cdots \; \rho_{K}]$ and $\boldsymbol{\rho}_{-k}=[\rho_1 \;\cdots \; \rho_{k-1}\; \rho_{k+1} \;\cdots \;\rho_{K}]$, for the sake of further developments. 
Let introduce a binary random variable $\Theta_k(\tau)$, which is equal to one when RB $\tau$ in station $k$ frame is granted to a UE, and equal to zero otherwise. 
We will assume that the probability for any station $k$ to receive data on RB $\tau$ depends only on $\rho_k$, and is equal to $\mathbb{E}_{\tau}[\Theta_k(\tau)]=\rho_k$.

Backhaul link buffer queue is also modeled as a multi-class M/G/1/PS (Fig.~\ref{fig:queues}). Each flow in the backhaul queue belongs to a class $k \in \{ 1\cdots n_{RN}\}$, according to the transmitting RN. The model shown in Fig.~\ref{fig:queues} thus allows for a hierarchical flow level analysis. 
\begin{figure}
 \centering
\includegraphics[width=\linewidth]{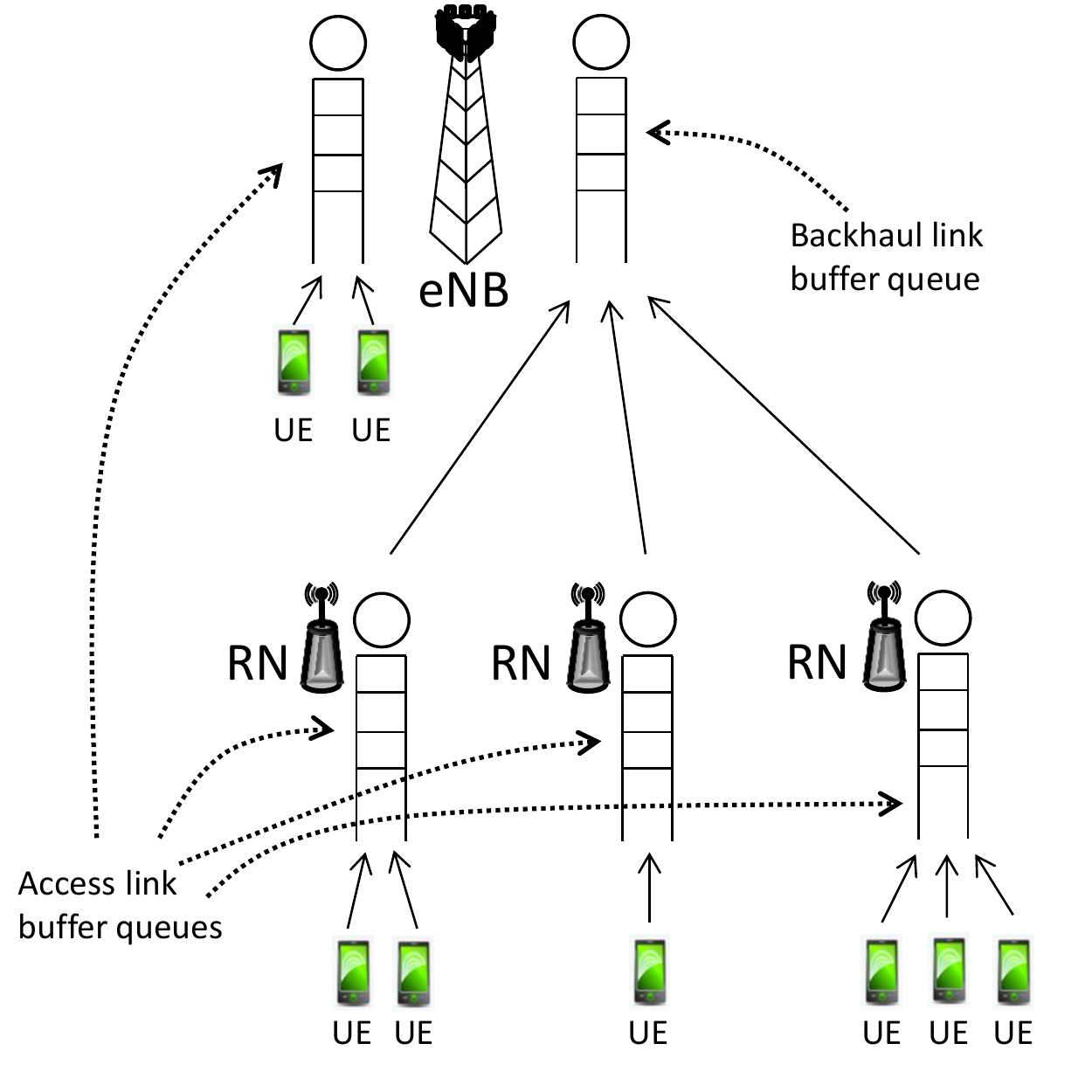} 
 \caption{Queuing model.}
\label{fig:queues} 
\end{figure} 

\subsection{Propagation Model}
Consider a UE transmitting on RB $\tau$, and located on $s$. We denote with $T(s)$ its transmit power, and with $G_k(s)$ the component of the channel gain towards station $k$ due to distance dependent attenuation and shadowing. We assume $G_k(s)$ constant on all RBs\footnote{The underlying assumption is that shadowing does not significantly change in time, for a fixed location. This is consistent with, e.g., the measurements in \cite{Kurose}.}. The power $P_k(s)$ received by $k$ from the considered UE is thus obtained as: $P_k(s,\tau)=T(s)G_k(s)\nu(s,\tau)$, 
where $\nu(s,\tau)$ is the variable component due to the fast fading. We adopt here a block Rayleigh fading model \cite{block_fading}, i.e., fast fading is constant on a RB, and fast fading realizations associated to any two distinct RBs or locations are independent of one another. The Full Compensation Power Control (FCPC) scheme adopted by the LTE standard \cite{36.814} is chosen to determine $T(s)$. We thus have:
$T(s)=\min\{ T_{\max} , \frac{\bar{P}}{G_k(s)} \}, \;s \in \mathcal{A}_k$,
where $T_{\max}$ is the UE maximum transmit power and $\bar{P}$ is a target received power broadcast by $k$ to all UEs.
We assume that $\bar{P}$ has the same value $\bar{P}{=}\bar{P}_{eNB}$ for all eNBs, and $\bar{P}{=}\bar{P}_{RN}$ for all RNs, for the sake of simplicity.    

\subsection{SINR model}\label{subsec:sinrmodel}
Consider a UE located in $s$ and a RB $\tau$. We define the instantaneous SINR $\gamma_k(s,\tau)$ experienced by station $k$ for the considered UE as:
\begin{equation}\label{eq:sinrdef}
 \gamma_k(s,\tau)=\frac{P_{k}(s,\tau)}{\sum_{j \in \mathcal{K}, j \neq k} \Theta_j(\tau) P_{k}(f_j(\tau),\tau) +N}, 
\end{equation}
where $f_j(\tau) {\in} \mathcal{A}_j$ yields the position of the UE scheduled by station $j$ on RB $\tau$, and $N$ is the thermal noise power. In the following, we drop the dependency of all variables on $\tau$, for the sake of simplicity. Let $p_{\gamma,k}(s,z), z \geq 0$ be the PDF of $\gamma_k(s)$. Following \cite{kostic, fischione}, we approximate $p_{\gamma,k}(s,z)$ by a lognormal distribution, i.e., $\gamma_k(s){\sim} \text{Log}{-}\mathcal{N}(\mu_{\gamma,k},\sigma_{\gamma,k})$ (see App.~\ref{appendix:momentmatch} for the derivation of $\mu_{\gamma,k}$ and $\sigma_{\gamma,k})$.



\subsection{Scheduling Policy}
Let assume that stations have perfect Channel State Information (CSI) about all served UEs and consider a given frame of station $k$. We label the UEs served by $k$ during the considered frame with an index $i=1\cdots U$, while their locations are denoted by $s_1 \cdots s_U$ respectively. Every station implements a {\it Maximum Quantile Scheduling} \cite{bonaldsched,MQS2} (MQS) among served UEs. 
The MQS scheduler of station $k$ assigns RB $\tau$ to the UE $i$ that maximizes its instantaneous SINR $\gamma_k(s_i,\tau)$, with respect to $\gamma_k(s_i,\tau-h),h=1\cdots W-1$ over a window of $W$ RBs (more details can be found in App.~\ref{appendix:MQS}).
This policy has the property of being fair in RBs allocation between UEs, while maintaining a good throughput performance \cite{MQS2}. Moreover, \cite{MQS2} shows that the performance degradation introduced by an imperfect estimation of the SINR distribution can be lower than that incurred in practical implementations of other popular scheduling algorithms. The MQS scheduling differs from the well-known PF scheduler \cite{KellyPF}, because the latter allocates each radio resource to the user maximizing its scaled SINR $\gamma_k(s_i,\tau)/\bar{\gamma}_k(s_i)$ on $\tau$, where $\bar{\gamma}_k(s_i)$ represents 
the average SINR experienced by $k$ for user $i$ over the last $W$ RBs. 
However, for large values of $W$ its behavior approximates that of a PF scheduler (see e.g. \cite{bonaldsched}). 

Now, consider the distribution $\pi_{\gamma,k}(s_i,z)dz = \mathbb{P}(z {\leq} \gamma_k(s_i) {\leq} z{+}{dz}\hspace{2 pt}| \hspace{2 pt} f_k{=}s_i), s_i \in \mathcal{A}_k, z \geq 0$ of $\gamma_k(s_i)$ {\it conditioned to} $f_k{=}s_i$, i.e., the PDF of the SINR of a {\it scheduled} UE. Contrary to $p_{\gamma,k}(s_i, z)$, this distribution is not lognormal.
\begin{theorem}\label{th:pi}
The distribution $\pi_{\gamma,k}(s_i,z)$ can be expressed as:
\begin{equation}
  \pi_{\gamma,k}(s_i,z)=p_{\gamma,k}(s_i,z) \displaystyle \sum_{n=1}^{W} Q_n(i,z) T_k(W,n),
\end{equation}
where 
\begin{eqnarray}
Q_n(i,z)  &=& \binom{W-1}{n-1}\left( \int_{0}^{z} p_{\gamma,k}(s_i,t)dt\right)^{W-n} \times \notag \\ & & \left(1- \int_{0}^{z} p_{\gamma,k}(s_i,t)dt \right)^{n-1}, \notag \\
T_k(W,n)  &=& \frac{ W^2(1-\rho_k)} { \left(W-\rho_k(W-n)\right)^2}. \notag
\end{eqnarray}
\end{theorem}
\begin{proof}
See App.~\ref{appendix:MQS}.
\end{proof}

\subsection{UEs Physical Data Rate}
We define the function $C(\gamma)>0$, yielding the physical data rate achieved, when the receiver experiences SINR $\gamma$ (or equivalently the throughput of a user, if it were alone in the serving area of its serving station). We assume that $C(\gamma)$ is non decreasing in $\gamma$.

%% file: perf_eval.tex
\def\dsp{\displaystyle}

\def\ie{\mbox{ {\em i.e.} }} 
\def\st{\mbox{ {\em s.t.} }} 

\def\st{\mbox{ s.t. }} 
\def\resp{\mbox{resp.~}}

\def\Demo{
\mbox{\em{Proof: }}}
\def\QED{\ \ \hfill QED \enspace.}
\def\ra{\rightarrow}
\def\Ra{\Rightarrow}
\def\lra{\leftrightarrow}
\def\Lra{\Leftrightarrow}
\def\cond{\ |\ }
\def\Xt{X^t}
\def\alphait{\alpha_i(t)}
\def\alphajt{\boldsymbol{\alpha}_{-i}(t)}
\def\Pt{\mathbb{P}_{\boldsymbol{\alpha}}}
\def\alphai{\alpha_i}
\def\alphaj{\boldsymbol{\alpha}_{-i}}
\def\ui{u_i}
\def\uj{\boldsymbol{u_j}}
\newcommand{\bbbe}{{\mathbb{E}}} 
\def\bbbone{{\mathchoice {\rm 1\mskip-4mu l} {\rm 1\mskip-4mu l} {\rm
1\mskip-4.5mu l} {\rm 1\mskip-5mu l}}} 
\def\charact#1{\bbbone_{#1}} 
\def\Esp#1{\bbbe\, \left[\, {#1} \, \right]} 
\def\Espt#1{\bbbe_{\boldsymbol{\alpha}}\, \left[\, {#1} \,\right]} 
\def\Esptt#1#2{\Espt{#2}} 

In this section, we define and derive our performance parameters in terms of delay and energy consumption per transmitted bit. We first obtain the expressions of access and backhaul link loads, and then we use them to get flow transmission delays. 

\subsection{Access Link Load} \label{subsec:accessload}

Let focus on station $k$, and recall that $\varrho_k(s), s \in \mathcal{A}_k$ is the contribution of $s$ to the load $\rho_k$ of station $k$. 
We have the following Lemma:
\begin{lemma}
The local contribution $\varrho_k(s),s \in \mathcal{A}_k$ to the load of the access link buffer queue of station $k$ is expressed by:
\begin{equation} \label{eq:rhoks}
 \varrho_k(s)=\frac{1}{1-\beta}\mathbb{E}_{\tau}\left[ \frac{\bar{\omega}\phi(s)}{C(\gamma_k(s,\tau))} \middle| f_k(\tau)=s \right].
\end{equation}
\end{lemma} 
\begin{proof}
The contribution to the load is expressed as the ratio of the traffic density $\bar{\omega}\phi(s)$ in $s$ to the uplink rate achieved by a scheduled UE in $s$. The term $(1-\beta)$ takes into account the quota of the frame RBs dedicated to the backhaul link.  
\end{proof}

\begin{theorem} \label{th:rho}
The load $\rho_k$ of the access link of station $k$ can be expressed as:
\begin{equation}\label{eq:rho_logn}
    \rho_k=\frac{\bar{\omega}}{1-\beta}  \int_{\mathcal{A}_k} \phi(s)  \int_{0}^{\infty}\frac{\pi_{\gamma,k}(s,z)}{C(z)}{dz}{ds}.
\end{equation}
\end{theorem}
\begin{proof}
The load $\rho_k$ is the sum of the contributions $\varrho_k(s)$ over $\mathcal{A}_k$: 
$\rho_k= \int_{\mathcal{A}_k} \varrho_k(s){ds}$. 
Now, the term $\mathbb{E}_{\tau}[ 1/C(\gamma_k(s,\tau))  | f_k(\tau)=s ]$ in (\ref{eq:rhoks}) can be expressed according the law of total probability, obtaining (\ref{eq:rho_logn}).
\end{proof}

Note first that if all UEs transmit at power $T(s)<T_{\max}, \forall s \in \mathcal{A}_k$, then $\pi_{\gamma,k}(s,z)=\pi_{\gamma,k}(z), \;\forall s \in \mathcal{A}_k$ and equation (\ref{eq:rho_logn}) reduces to: $\rho_k=\frac{\bar{\omega}\Phi_k}{1-\beta} \mathbb{E}_{\gamma}\left[\frac{1}{C(\gamma)}\right]$. $\Phi_k$ can thus be seen as the {\it effective area} covered by $k$, in terms of traffic density.

Note then that $\rho_k$ is a function of $\boldsymbol{\rho}_{-k}$, as the load depends on the interference from other stations. 
We define the operator $F_k(\boldsymbol{\rho}_{-k})=\bar{\omega}/(1-\beta)  \int_{\mathcal{A}_k} \phi(s)  \int_{0}^{\infty}\pi_{\gamma,k}(s,z)/C(z){dz}{ds}$, which yields the value of $\rho_k$ corresponding to a given $\boldsymbol{\rho}_{-k}$.

\begin{lemma} \label{lm:finite}
The operator $\boldsymbol{F}(\boldsymbol{\rho})=(F_1(\boldsymbol{\rho}_{-1}) \cdots F_{K}(\boldsymbol{\rho}_{-K})  )$ maps $\boldsymbol{\rho}$ to a finite $K$-dimensional interval.
\end{lemma}
\begin{proof}
For all stations $k \in \mathcal{K}$, the corresponding $F_k(\boldsymbol{\rho}_{-k}) $ is increasing with respect to every coordinate, as an increase in the load of any station $j\neq k, j\in \mathcal{K}$ produces an increase in the interference received at $k$.
We define $\forall k, \rho_k^{\max}=F_k(1 \cdots 1)$, which represents an upper bound for the load of $k$. Then, $\forall k, \forall \rho_k, F_k(\boldsymbol{\rho}_{-k})=F_k(\max\{ \boldsymbol{1} , \boldsymbol{\rho}_{-k} \})\leq F_k(1 \cdots 1)=\rho_k^{\max}$.
Hence, 
\begin{equation}\label{eq:limited}
 \boldsymbol{F}( \boldsymbol{\rho}) \in \prod_{k\in \mathcal{K}} [0 , \rho_k^{\max}], \;\rho_k^{\max}\in \mathbb{R}^+, \;\forall \hspace{1 pt} k \in \mathcal{K}.
\end{equation} 
\end{proof}
\begin{theorem}
There exists at least one $\boldsymbol{\rho}^*$ such that $\boldsymbol{F}(\boldsymbol{\rho}^*)=\boldsymbol{\rho}^*, \boldsymbol{\rho}^* \in \prod_{k\in \mathcal{K}} [0 , \rho_k^{\max}]$.
\end{theorem}
\begin{proof}
This follows from Lemma \ref{lm:finite} and (\ref{eq:limited}): if $\boldsymbol{F}$ is a continuous mapping from an $K$-dimensional closed interval to itself, then the Brouwer's fixed point theorem guarantees the existence of at least one fixed point.
\end{proof}

It is not feasible in general to draw any conclusion about the uniqueness of the fixed point. However, following the approach of \cite{Rong11}, we start from a single cell without interference ($\boldsymbol{\rho}=\boldsymbol{0}$) and iterate function $\boldsymbol{F}$, so as to model a scenario of increasing traffic.  

\subsection{Backhaul Link Load} \label{subsec:bhload}
We now focus on the backhaul link queue of the eNB of cell $c$. Let $j \in \mathcal{K}_c$ be the index of the eNB and $\rho_{BL,j}$ the load of the queue. The probability for a given RN $k$ served by $j$ to be selected for backhaul transmission is equal to $\Phi_k/\sum_{k \in \mathcal{K}_c}\Phi_k$, where $\Phi_k = \int_{\mathcal{A}_k}\phi(s){ds}$. Hence, $\rho_{BL,j}$ can be written as:
\begin{equation}
 \rho_{BL,j}=\frac{\bar{\omega}}{\beta} \displaystyle \sum_{k=1}^{n_{RN}} \frac{\Phi_k}{R_{BL,j,k}}, 
\end{equation}
where the inverse of $R_{BL,j,k}$ is the average of the inverse of the rate on the backhaul link between RN $k$ and eNB $j$, which depends on propagation conditions, and is equal to:
$R_{BL,j,k}=1/\mathbb{E}_{\tau}[ \frac{1}{C\left( \gamma_{BL,j,k}(\tau) \right)} ]$,
where $\gamma_{BL,j,k}(\tau)$ denotes the SINR on the backhaul of RN $k$ served by eNB $j$, during RB $\tau$.  
A detailed derivation of $R_{BL,j,k}$ is presented in App.~\ref{appendix:back}. 

\subsection{Flow Average Transmission Delay}
The flow average transmission delay is an effective parameter to measure network uplink performance. It is the sum of the average access delay and the average backhaul delay. In this section, we derive its formulation. Consider a user located in $s$, served by RN $k$ with donor eNB $j$. We denote by $D_k(s)$ the average access delay of a flow in $s$, by $D_k$ the average delay of a flow served by $k$ and by $D_{BL,j,k}$ the average delay on the backhaul link.


\begin{lemma}\label{lm:da}
The average access delay in $s$ is expressed by:
\begin{equation}\label{eq:da}
    D_k(s) =   \frac{\xi }{1-\rho_k}   \int_{0}^{\infty}\frac{\pi_{\gamma,k}(s,z)}{(1-\beta)C(z)}{dz}.
\end{equation}
\end{lemma}
\begin{proof}
The delay to transmit a bit of information is the inverse of the UE rate, multiplied by the number $x_k$ of UEs served by the same station during a frame (because each UE is scheduled on a fraction $1/x_k$ of the RBs with MQS). 
The average transmission delay for the whole flow is hence given by:
$D_k(s)= \xi \mathbb{E}_{\tau}\left[ \frac{x_k}{(1-\beta)C(\gamma_k(s,\tau))}\middle | f_k(\tau)=s \right]$. 
Considering that $\mathbb{P}(x_k=U|x_k>0)=\rho_k^{U-1}(1-\rho_k)$, the law of total probability can be used to average $D_k(s)$ over $x_k>0$:
$D_k(s)=\displaystyle \sum_{U=1}^{\infty}\mathbb{P}(x_k=U)\frac{\xi}{1-\beta} U \mathbb{E}_{\tau}[ \frac{1}{C(\gamma_k(s,\tau))} | f_k(\tau)=s]$,
obtaining expression (\ref{eq:da}).
\end{proof}

\begin{corollary}
The average access delay for a UE served by $k\in \mathcal{K}_{R}$ is:
\begin{equation}\label{eq:delay_1}
  D_k = \int_{\mathcal{A}_k}\frac{\phi(s)}{\Phi_k}D_k(s){ds}.
\end{equation}
\end{corollary}
\begin{proof}
The statement can be verified by means of the Little's law:
$D_k =   \frac{\mathbb{E}_{\tau}\left[  x_{k}(\tau)   \right]  }{\bar{\omega} \Phi_k }= \frac{\rho_k}{1-\rho_k} \frac{1}{\bar{\omega}\Phi_k}$.
Now, we can substitute (\ref{eq:rho_logn}) to the numerator and conclude the proof.
\end{proof}

\begin{lemma}
The average backhaul link delay is expressed by:
\begin{equation} \label{eq:dbljk}
  D_{BL,j,k}=\frac{\xi}{(1-\rho_{BL,j})\beta R_{BL,j,k}}, k\in \mathcal{K}_R,j\in \mathcal{K}_B.
\end{equation}
\end{lemma}
\begin{proof}
Similarly to what performed for Lemma \ref{lm:da}, we proceed by expressing $D_{BL,j,k}$ as:
$D_{BL,j,k}=\frac{\xi}{\beta R_{BL,j,k}} \mathbb{E}_{\tau}[  x_{BL,j}(\tau) | x_{BL,j}(\tau)>0 ]$ and apply the law of total probability, obtaining: 
\begin{eqnarray} 
            D_{BL,j,k}  &=& \!\! \frac{\xi}{\beta R_{BL,j,k}} \displaystyle \sum_{v=1}^{\infty} P \left(x_{BL,j}(\tau)\!=\!v \middle | x_{BL,j}(\tau)\!>\!0 \right) \times \notag \\
            && \mathbb{E}_{\tau}\left[ x_{BL,j}(\tau) \middle |  x_{BL,j}(\tau)=v \right],\notag 
\end{eqnarray} 
where $x_{BL,j}$ is the number of flows in the backhaul queue of $j$ during the considered frame. Conventionally, we set $D_{BL,j,j}=0$. From this equation, we obtain (\ref{eq:dbljk}).
\end{proof}

\begin{proposition}
The average transmission delay $\bar{D}_c$ for an uplink flow transmitted in cell $c$ is equal to:
\begin{equation} \label{eq:delay}
 \bar{D}_c=\displaystyle \sum_{k \in \mathcal{K}_c}  \frac{\Phi_k}{A_k} (D_k+D_{BL,j,k}).
\end{equation}
\end{proposition}

\subsection{Energy Consumption per Transmitted Bit}
Let $\epsilon(s)$ be the energy consumed by a UE located in $s\in \mathcal{A}_{k}$ for transmitting one bit. This metric is sometimes called Energy Consumption Rating in the literature \cite{Suarez12, Han11}. We have:
\begin{eqnarray}
  \epsilon(s)&=&\mathbb{E}_{\tau} \left[ \frac{T(s)}{C(\gamma_k(s,\tau))}\middle | f_k(\tau)=s \right] \notag \\
  &=&T(s)\displaystyle \int_{0}^{\infty} \frac{\pi_{\gamma,k}(s,z)}{C(\gamma_k(s,z))}{dz},s\in \mathcal{A}_k
\end{eqnarray}
which can be also expressed as a function of $D_k(s)$:
$\epsilon(s)=T(s) D_k(s)(1-\rho_k)(1-\beta), s\in \mathcal{A}_k$.
The {\bf average uplink energy consumption per bit} $\Pi$ associated to UEs in cell $c$ is given by:
\begin{eqnarray} \label{eq:ee}
 \Pi &=&\frac{1}{A_c}\displaystyle \int_{\mathcal{A}_c} \phi(s)\epsilon(s){ds} \notag \\
 &=& \frac{1-\beta}{A_c}\displaystyle \sum_{k \in \mathcal{K}_c} (1-\rho_k) \label{eq:ee}\int_{\mathcal{A}_k} T(s)D_k(s) \phi(s){ds}.
\end{eqnarray}

%% file: optimization.tex
\def\config#1{\boldsymbol{#1}}
\def\objective{{\Pi}} 
\def\constraint{D_c} 
\def\constraintmax{D_{\mbox{{\tiny max}}}}
\def\prior{V}
\def\multiplier{\mu}
\def\augmented{{\cal F}}
\def\feasable{\tilde{\Omega}}

\def\slack{t}

\def\Real{\mathbb{R}}
\def\st{\mbox{ s.t. }} 
\def\ie{\mbox{\em i.e., }} 
\def\wrt{\mbox{\em wrt. }} 

\def\ul{\underline}
\def\ra{\rightarrow}

\def\charact{\bbbone}
\def\dsp{\displaystyle}
\def\cond{\ |\ }
\def\Ra{\Rightarrow}

\def\ra{\rightarrow}
\def\Ra{\Rightarrow}
\def\lra{\leftrightarrow}
\def\Lra{\Leftrightarrow}

\def\ie{\hspace*{-1mm}\mbox{, {\em i.e.,} }} 
\def\dsp{\displaystyle}

\def\Pr{\mbox{Pr}} 
\def\bbbp{\mathbb{P}}
\def\Pr{\bbbp}

\def\proposal{r} 
\def\acceptance{\Xi}

In this Section, we discuss the optimization of UEs uplink energy efficiency. 
The proposed algorithm aims at minimizing the UE average energy consumption per transmitted bit in (\ref{eq:ee}), by optimizing $\bar{P}$, $\mathcal{B}$ and the RN placement in the cell of interest, for a given $\bar{\omega}$, $n_{RN}$ and traffic spatial profile $\phi(s)$. Our optimization is constrained to respect a maximum tolerable value of $\bar{D}_c$ in (\ref{eq:delay}), so as to take into account the users experienced quality of service. 

RN placement is usually optimized with respect to the {\it downlink} performance of the network because of its larger traffic volume. On the contrary, we consider here the {\it uplink} performance because we aim at finding upper bounds on the possible gains that an operator can achieve in terms of energy and/or average delay on this link. Note however that with the growing traffic related to multimedia content sharing, social networks and other peer-to-peer applications, uplink and downlink traffics tend to be more balanced\footnote{A joint uplink and downlink RN placement optimization is left for further work.}. 
  
\subsection{Problem Statement}
We assume that RNs can be placed on a discrete and finite grid of candidate sites inside cell $c$\footnote{This assumption is consistent with real deployment scenarios, where normally only a certain number of locations inside the cell are available for RN installation \cite{bulakciRSP}.}, while the positions of RNs outside cell $c$ are assumed to be already set by the network operator, and not modifiable during the optimization. Position of a given RN $i$ is denoted with $z(i)$. 
Moreover, we define $\boldsymbol{\bar{P}}=[\bar{P}_{eNB} \hspace{5 pt}\bar{P}_{RN}]$.
Target received powers $\bar{P}_{eNB}$ and $\bar{P}_{RN}$ can vary between a minimum and a maximum value, denoted with $\bar{P}_{\min}$ and $\bar{P}_{\max}$ respectively. Similarly, we assume $\mathcal{B}_{\min} \leq  \mathcal{B} \leq \mathcal{B}_{\max}$ (see Section \ref{system_model}).
Now, the {\it configuration} $\boldsymbol{x}$ of the network is defined as the set of positions of the $n_{RN}$ RNs in cell $c$, plus the adopted $\boldsymbol{\bar{P}}$ and $\mathcal{B}$:
\begin{align}
  & \boldsymbol{x} : \left\{  z(1), \cdots ,z(n_{RN}), \boldsymbol{\bar{P}}, \mathcal{B}  \right\}, \notag \\
  & z(i)\in \mathcal{A}_{c}, \hspace{2 pt} \forall \hspace{2 pt}i \in \{1\cdots n_{RN}\}, \notag \\
  & \bar{P}_{\min}\leq \bar{P} \leq \bar{P}_{\max}, \hspace{2 pt} \mathcal{B}_{\min} \leq  \mathcal{B} \leq \mathcal{B}_{\max}.
\end{align}
We name {\it configuration space} the set of all configurations, and denote it with $\Omega$.

Our problem is to minimize energy consumption (\ref{eq:ee}) under delay constraint (\ref{eq:delay}):
\begin{equation}
\min_{\config{x}} \objective(\config{x}) \label{objective:eq} 
\end{equation}
\begin{equation}
\mbox{s.t. } \bar{D}_c(\config{x}) \leq \constraintmax,\; \config{x} \in \Omega, \constraintmax \in \mathbb{R}^+, \label{constraint:eq}
\end{equation}
where constraint (\ref{constraint:eq}) restricts the domain of feasible solutions to the subspace: 
\begin{equation}
\feasable = \{ \config{x} \in \Omega \st \constraint(\config{x}) \leq \constraintmax \} 
\end{equation}
Recall that the computation of the station loads via the fixed point iteration of Section~\ref{subsec:accessload} is required for the evaluation of the delay.

Now, the problem (\ref{objective:eq}) is in general non-convex, and the cardinality of $\Omega$, especially for high $n_{RN}$, makes it intractable via exhaustive search. Hence, we rely on stochastic optimization algorithms, and propose a customized version of the well-known Simulated Annealing algorithm to solve (\ref{objective:eq}). In the following, we first briefly recall the generic SA, and then introduce our version.

\subsubsection{Generic SA Algorithm} 
The SA is a metaheuristic aimed at solving large non-convex problems, which has been first proposed by Metropolis \cite{Metropolis1953} and then applied on a wide range of optimization problems (see, e.g., \cite{Kirkpatrick1983,Keung2010}). The literature shows that the SA is an effective algorithm, if its parameters are appropriately tuned (see, e.g., \cite{Roubini99}). Let $\augmented(\config{x})$ denote the {\it energy}\footnote{Not to be confused with the energy in J involved in the energy consumption (\ref{eq:ee}).} associated with configuration $\config{x}$, and consider the problem of minimizing $\augmented(\config{x})$ over the configuration space $\Omega$. 
The SA explores only a subset $\Psi \subset \Omega$ of the configuration space, where usually $|\Psi| \ll |\Omega|$, and is able to find the optimal configuration by means of an {\it appropriate} selection of the analyzed configurations. At temperature $T$, the algorithm proceeds by assigning to each configuration $\config{x}$ an exponential probability, given by:
\begin{equation} \label{eq:probaexp}
     \mathbb{P}_T(\config{x})= \frac{ \exp(-\frac{\augmented(\config{x})}{T})}{Z_T}, \forall \config{x} \in \Omega,
\end{equation} 
where $Z_T$ is a normalizing constant. 
Hence, the solution is the $\config{x}$ that maximizes $\mathbb{P}_T(\config{x})$. According to the Metropolis-Hastings variant \cite{Hastings1970} of the SA, the set $\Psi$ of configurations to be analyzed is determined according to the following procedure: 
\begin{itemize}
	\item At step $m=0$ a configuration $\config{x}(0) \in \Omega$ is arbitrarily selected and designated as {\it current solution} $\config{x}^*(0)$ for step $0$: $\config{x}^*(0) \leftarrow \config{x}(0) \in \Omega$.
	\item At any step $m \geq 0$: assume $\config{x}^*(m-1)$ is the current solution for step $m-1$.
The SA will first draw a new {\it candidate} solution $\config{x}'(m) \in \Omega$ for step $m$, which complies with the given {\it proposal law} $\proposal(\config{x}^*(m-1) \rightarrow \config{x}'(m))$ (defined by the algorithm designer). 
Then, assuming that $r$ is symmetric, $\config{x}'(m)$ will be accepted as current solution for step $m$ with probability:  
\begin{eqnarray}
  \lefteqn{\mathbb{P}_{T_m}(\config{x}^*(m) \leftarrow \config{x}'(m))=} \notag \\ 
  &&\min\left(1,\ \exp \left( - \dfrac{\augmented(\config{x}'(m)) - \augmented(\config{x}^*(m-1))}{T_m} \right)\right),  \notag
\end{eqnarray}
where $T_m$ is a parameter called {\it temperature}, which decreases to zero {\it slowly enough}  as $m\rightarrow+\infty$, and is such that $T_m \geq \frac{T_0}{1 + \log (m+1)}$.
If $\config{x}'(m)$ is not accepted, then $\config{x}^*(m) \leftarrow \config{x}^*(m-1)$. In most practical implementations of the SA, the temperature is updated following a law of the kind $T_m = T_0h^m, h<1$, where $h$ is {\it close} to 1.
\end{itemize}

\subsubsection{Proposed Exterior Search Approach for SA}
We seek for a minimizer of $\objective( \config{x})$ over $\tilde{\Omega}$, where the cardinality of the feasible configurations space  depends on the value of the constraint: $|\tilde{\Omega}|=f(D_{\max})$. In particular, $f$ is an increasing function of $D_{\max}$, because the number of feasible configurations increases as the constraint loosens. 

When $D_{\max}$ is small, we can reasonably expect to have $|\tilde{\Omega}|\ll |\Omega|$. This may reduce the effectiveness of the SA, as the energy surface that the algorithm explores could be fragmented into many isolated regions, some of which are unreachable for the algorithm. Moreover, we can expect the optimal solution to tightly respect the constraint, i.e., to lie close to the border of $\tilde{\Omega}$ \cite{Siedlecki89}. Thus, the standard SA may not appropriately cover the region where the solution is, as many of its configurations are non feasible.   

This problem can be solved by means of an {\it exterior search} approach \cite{Schwefel95,kim-course}. It consists in extending the search to configurations outside $\feasable$, while adding a {\it penalty} to the energy of those configurations that violate the constraint. This method presents analogies with the Lagrangian relaxation method (see the classic works \cite{HK70,Fisher73}), and it is regarded as a powerful instrument to deal with constrained optimization problems (see, e.g., \cite{Michalewicz95,Coit96}). This idea is illustrated in Fig.~\ref{ext-penalty:fig}, where we see on a fictitious example how exterior search can reduce the path to the optimal configuration. 
\begin{figure}[t]
\vspace{5mm}
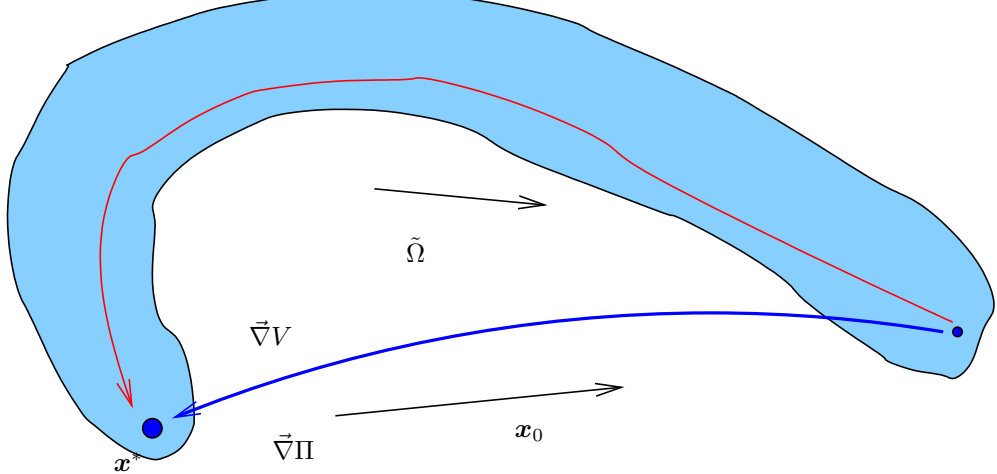
\caption{\label{ext-penalty:fig}
Exterior search principle: the search from the initial configuration $\config{x}_0$ to the optimal configuration $\config{x}^*$ is extended to configurations outside $\tilde{\Omega}$ (in red: interior penalty, in blue: exterior penalty).
}
\end{figure}

We thus adopt an exterior search approach, and reformulate (\ref{objective:eq}) as:
\begin{equation}
\min_{\config{x}}  \augmented(\config{x}),  \config{x} \in \Omega, \label{eq:optim_reformulated}
\end{equation}
\begin{equation}
\mbox{where }\augmented(\config{x}) = \objective(\config{x}) 
+  \mathbb{1}_{D_c(\config{x})>D_{\max}}V(\config{x}, D_{\max}), \notag
\end{equation}
and where $\prior(\config{x})$ is a function $\Omega \mapsto \Real$ accounting for constraint (\ref{constraint:eq}) and named {\it penalty function}, which plays a crucial role in the optimization \cite{kim-course}.

Penalty functions are classified into {\it static}, {\it dynamic}, and {\it adaptive}. A static penalty function does not change during the optimization. It can be represented, e.g., by a fixed constant added to the energy of those configurations $\config{x}\not \in \tilde{\Omega}$, or by a function proportional to the Euclidean distance of the considered configuration to the feasible region \cite{Richardson90}. On the contrary, dynamic penalty functions can be adjusted in accordance to the progress of the optimization. A common approach is to gradually increase the penalty with the number of explored configurations \cite{Joines94}, so as to guarantee the convergence of the optimization towards a feasible solution. Finally, an adaptive penalty function (see, e.g., \cite{Coit96,Bean92}) considers further aspects of the search, such as steering the algorithm towards regions of the energy surface which are deemed promising for the search.    

We propose the following exterior penalty function, which is dynamic and adaptive at the same time: 
\begin{equation}\label{eq:adoptedpen}
    V(\config{x}, D_{\max}) = \alpha (m-1)\frac{ \objective(\config{x})}{D_{\max}}(D_c(\config{x})-D_{\max}),
\end{equation}
where $\alpha$ is a constant and $m$ is the SA step.
The benefits of using function (\ref{eq:adoptedpen}) are manifold. First, the penalty is proportional to the violation of constraint $D_{\max}$, favoring the exploration of configurations which are out of $\tilde{\Omega}$ but {\it close} to it, while penalizing more those which are {\it far}. At the same time the penalty is independent of the adopted value of $D_{\max}$, depending rather on the percentage of exceeding delay, with respect to $D_{\max}$. This represents an important feature when dealing with constraints of a different order of magnitude, compared to $\objective(\config{x})$. Moreover, the penalty is proportional to $\objective(\config{x})$, so as to favor the exploration of configurations outside $\tilde{\Omega}$ when they appear to be {\it promising}, i.e., their energy is sensibly lower to that of the current solution (hence the adaptive nature of (\ref{eq:adoptedpen})). Finally, acting on $\alpha$ we can modify the percentage of acceptance for configurations outside $\tilde{\Omega}$. We use here the works of Geman and Robini \cite{geman-graffigne, robini-bresler-magnin}: Choosing $\alpha(m-1) > \log(m)$ is sufficient to guarantee convergence of the SA algorithm towards a feasible minimizer. App.~\ref{app:gibbs} explains why the convergence property of the SA is maintained with such a penalty function. 

Surprisingly, stochastic constrained optimization based on exterior penalties has been seldom employed in wireless networks optimization. There are however three recent interesting references using this method \cite{Yang11, Saifullah14, Azim12}. In this context, two important problems arise:
\begin{enumerate}
	\item The generic choice of the penalty term itself. To our knowledge, the most interesting analysis appears in \cite{Azim12}, using the so-called logarithmic barrier. Our approach is somewhat similar: since our penalty term is the product of the objective function $\Pi(\config{x})$ and the constraint $\constraint(\config{x}) - \constraintmax$, its gradient is a linear combination of the gradients of both previous expressions. Thus, in a deterministic, continuous framework, during a phase when the penalty is active ($\constraint(\config{x}) > \constraintmax$) it adapts the minimizing direction search to both the objective function and the constraint  opposite gradients (see Fig. \ref{ext-penalty:fig}). We expect the stochastic, discrete SA to behave in a somewhat similar way.
	\item The selection  of the Lagrangian multipliers and their evolution during the algorithm. Apart from the usual choice of constant parameters, an interesting development can be found in \cite{Yang11}: at each step of the SA algorithm, each penalty weight itself either increases if the associated constraint is violated, or decreases (in a geometric way) if this constraint has been satisfied during several previous steps. In \cite{Saifullah14} all penalty weights increase regularly.
\end{enumerate}

To the best of our knowledge, our approach is the first one to use stochastic constrained optimization for station placement in wireless networks. The originality of our work also lies in the multiplicative penalty function and the theoretical choice of the penalty coefficient $\alpha(m-1)$.


\subsubsection{Initial Temperature Choice} 
Following \cite{relplac}, we address the problem of finding a good value for the initial temperature $T_0$ in each considered optimization problem via dichotomic search during a series of preliminary runs of the algorithm.   


%% file: extern-penalty.pdf_t
\begin{picture}(0,0)%
\includegraphics[width=0.8\linewidth]{./extern-penalty.pdf}%
\end{picture}%
\setlength{\unitlength}{4144sp}%
\begingroup\makeatletter\ifx\SetFigFont\undefined%
\gdef\SetFigFont#1#2#3#4#5{%
  \reset@font\fontsize{#1}{#2pt}%
  \fontfamily{#3}\fontseries{#4}\fontshape{#5}%
  \selectfont}%
\fi\endgroup%
\begin{picture}(4000,1400)(0,0)
\put(2400,1200){$\tilde{\Omega}$}
\put(3050, 150){${\config{x}}_0$}
\put(650, -50){${\config{x}}^*$} 
\put(1600,20){$\vec{\nabla}\Pi$}
\put(1460,700){$\vec{\nabla}V$}
%
\end{picture}%

%% file: results.tex
\subsection{Considered Scenario and SA Parameters}
The proposed theoretical framework is here applied on a test scenario drawn according to \cite[case 1]{3gpprelayperf} for what concerns propagation, shadowing and stations transmit power. System bandwidth is $10$~MHz. The network is composed by one central eNB surrounded by 6 eNBs regularly distributed around it on a circle of ray $500$~m. We optimize RN placements in the central cell only, while RNs in the surrounding cells are assumed to be regularly distributed on a circle of radius 160 m around their donor eNB. All stations have omnidirectional antennas. The same realization of shadowing (drawn according to \cite{3gpprelayperf}) is used for all simulations, so as to ensure that results be comparable. If not otherwise mentioned, we adopt a {\it uniform} traffic density spatial distribution, i.e., $\phi(s)$ is constant $\, \forall \, s {\in} \mathcal{A}_{nw}$, in order to draw general conclusions about energy efficiency. In Section~\ref{non_unif} we show the performance under a {\it non-uniform} traffic distribution. We set $\beta{=}0.1$. 

Function $C$ is approximated by means of the Modulation and Coding Scheme (MCS) indicated in \cite{36.942}, so as to take into account the effect of an upper-bounded capacity function, which is the typical case in real deployments. Also, fast fading on the backhaul link is not considered. This choice is justified by the assumption that RNs do not move, and optimization of their positions is performed on a long-term basis (see, e.g. \cite{salehFF}). 
We have assumed $\mathcal{B}_{\min}{=}0$ and $\mathcal{B}_{\max}{=}15 \hspace{2 pt}$~dB, with a step of 1 dB. The grid of candidate RN sites has a step of $50$~m.
    

The fixed point of $\boldsymbol{F}(\boldsymbol{\rho})$ is found by iteratively computing stations loads. At iteration $t$, probabilities $\pi_{\gamma,k}(z,s)$ are computed according to $\boldsymbol{\rho}(t-1)$, obtained at iteration $t-1$, and then fed into (\ref{eq:rho_logn}) to get $\boldsymbol{\rho}(t){=}\boldsymbol{F}(\boldsymbol{\rho}(t-1))$, starting from $\boldsymbol{\rho}(0)=\boldsymbol{0}$. Iterations are stopped when either $|\rho_k(t)-\rho_k(t-1)|{<}0.01, \forall k \in \mathcal{K}$, or $\exists  k : \rho_k(t)\geq 1, k \in \mathcal{K}$. In the latter case, the analyzed configuration is labeled as non-valid.  
Similarly, fixed point iterations are used to compute the loads on the backhaul link, after $\boldsymbol{\rho}$ is determined. We were not able to analytically prove the convergence of the fixed point iteration (by showing for example that $\boldsymbol{F}$ is a contraction mapping). We have however numerically observed that the iteration always converges in less than 10 iterations.

SA algorithm is implemented as described in Section \ref{optimization}. Once again, we emphasize the need to optimize network parameters and in particular RN locations in order to obtain upper bounds on the achievable performance. The algorithm yields a solution after 45 temperature steps. For each tested temperature value, the energy of a certain number of network configurations is calculated. This number varies according to $n_{RN}$, e.g., for $n_{RN}=2$, 400 network configurations are tested at each temperature step. 
For every optimization, the SA is repeated 4 times, and the best solution among the 4 obtained solutions is elected as a final result. The configuration corresponding to the final result is denoted with $\boldsymbol{x}^*$. 

\subsection{Numerical Results}
We denote with $\boldsymbol{x}^*_{0}$ the configuration with minimum energy when no RNs are deployed in the whole network, and with $\Pi_0(\boldsymbol{x}^*_{0})$ its corresponding energy. 
Unless otherwise mentioned, results are expressed in terms of the energy consumption ratio $\Pi (\boldsymbol{x}^*) / \Pi_0(\boldsymbol{x}^*_{0})$, so as to show the energy gain (or loss) resulting from RN deployment, and they are plotted against the normalized constraint $D_{\max}/D_0$, where $D_0$ is the value of $\bar{D}_c$ corresponding to $\boldsymbol{x}^*_{0}$, i.e., the average delay when no RNs are deployed. 
Hence, a point in the region $\Pi(\boldsymbol{x}^*) / \Pi_0(\boldsymbol{x}^*_{0}){<}1$ corresponds to an energy consumption per bit gain with respect to the network with no RNs, whereas a point in the region $D_{\max}/D_0{<}1$ corresponds to an average uplink transmission delay gain. 
All the curves that we obtained exhibit a hyperbolic-like shape, which follows from the nature of our constrained optimization. When $D_{\max}$ is large, $| \tilde{\Omega}| \approx |\Omega|$, the delay doesn't play any role in the optimization and so the energy consumption gain reaches its maximum. 
Instead, if the constraint is tight we have $| \tilde{\Omega}| \ll |\Omega|$, and it is unlikely that any configuration that performs {\it well} in terms of energy efficiency lies in $\tilde{\Omega}$. 

\subsubsection{Energy - Delay Trade-off}
Fig.~\ref{fig:uni_01} shows the trade-off between energy consumption and delay, for a varying number of RNs. Note how RN deployment can bring consistent UEs energy savings.
Let consider the region where $D_{\max}/D_0{>}1$. We observe that the energy consumption gain increases with $n_{RN}$. 
This is an expected result, as RNs reduce the average distance between UEs and their serving station. The corresponding increase in uplink interference is mitigated by a lower users transmit power. On the contrary, a lower number of RNs performs better when $D_{\max}/D_0<0.5$. This behavior is related to the cardinality $|\tilde{\Omega}|$ of the valid configurations set, which tends to decrease with $D_{\max}/D_0$, narrowing the space for network optimization and reducing the achievable energy consumption gains. In this case adding RNs is detrimental for energy efficiency, as more UEs served by RNs imply a higher backhaul delay, that further reduces $|\tilde{\Omega}|$.  

\subsubsection{Effect of Offered Traffic}
We analyze the effect of $\bar{\omega}$ on energy consumption in Fig.~\ref{fig:load_final}, where $\Pi(\boldsymbol{x}^*)$ and $D_{\max}$ are scaled with the constants $\Pi_0(\boldsymbol{x}_0^*(\bar{\omega}=5))$ and $D_{0}(\bar{\omega}=5)$ resp., obtained in a network with no RN and a traffic density $\bar{\omega}=5$~bits/s/m$^2$. This choice has been made in order to keep the same scaling constants for all curves. 
As we can see, the deployment of 1 RN is sufficient to reduce the average energy consumption of UEs to less than a half, compared to the eNBs-only network.
Now, let focus on the line $D_{\max}/ D_0(\bar{\omega}=5){=}1$ (same delay as in the eNB-only case). We note that adding a single RN can help the terminals to be more energy efficient without touching the quality of service and even if the load increases to $\bar{\omega}=20$~bits/s/m$^2$. Nonetheless, beyond approximately $\bar{\omega}=30$~bits/s/m$^2$, delay and interference negatively affect user performance. 


%
\begin{figure}
\centering
\begin{minipage}{.47\textwidth}
\centering
\includegraphics[width=\linewidth]{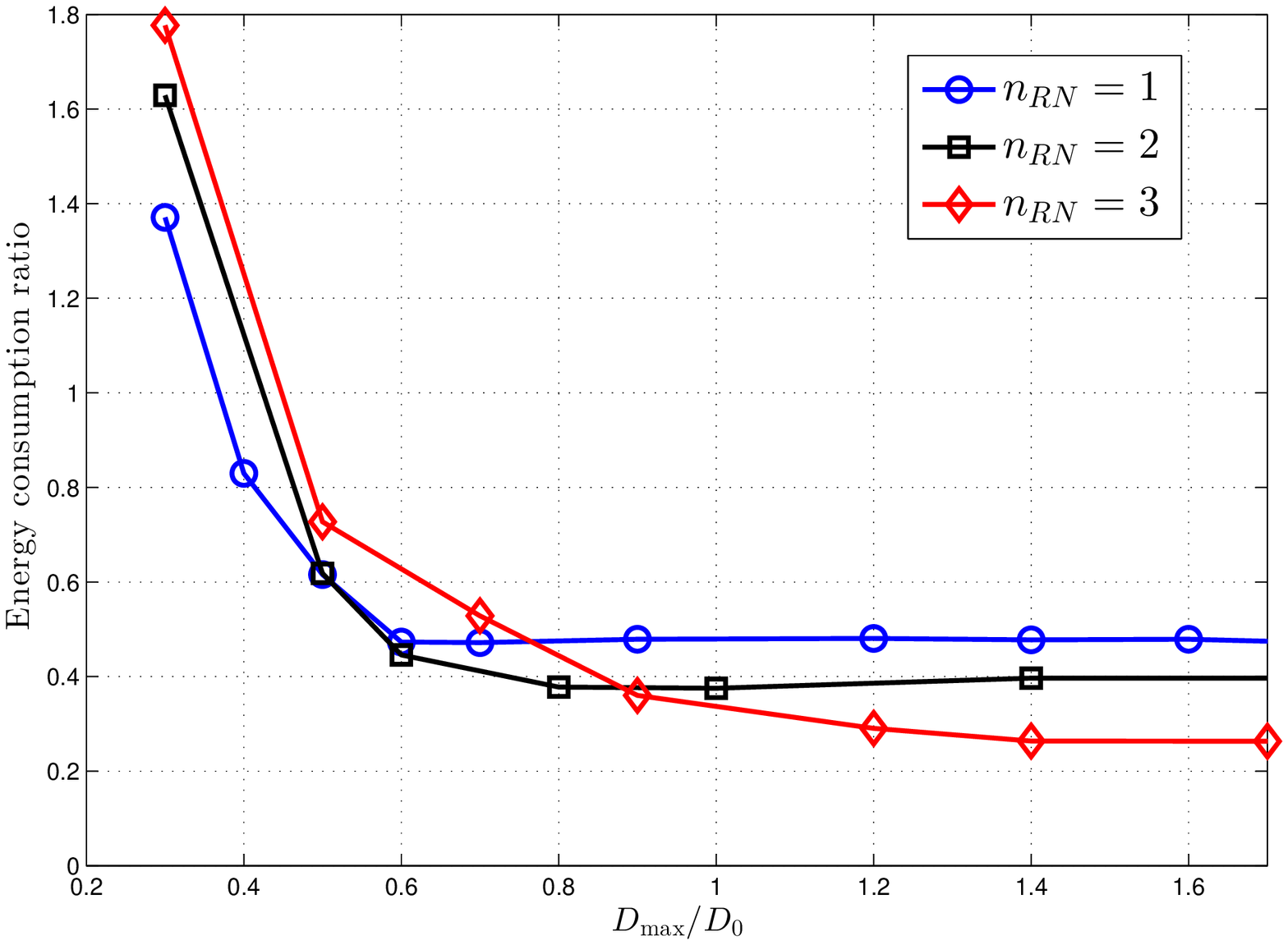} 
\caption{Energy consumption ratio $\Pi (\boldsymbol{x}^*) / \Pi_0(\boldsymbol{x}^*_{0})$ vs delay $D_{\max}/D_0$ ($\bar{\omega}=5$~bits/s/m$^2$).}
\label{fig:uni_01}
\end{minipage}\hfill
\begin{minipage}{.47\textwidth}
\centering
\includegraphics[width=\linewidth]{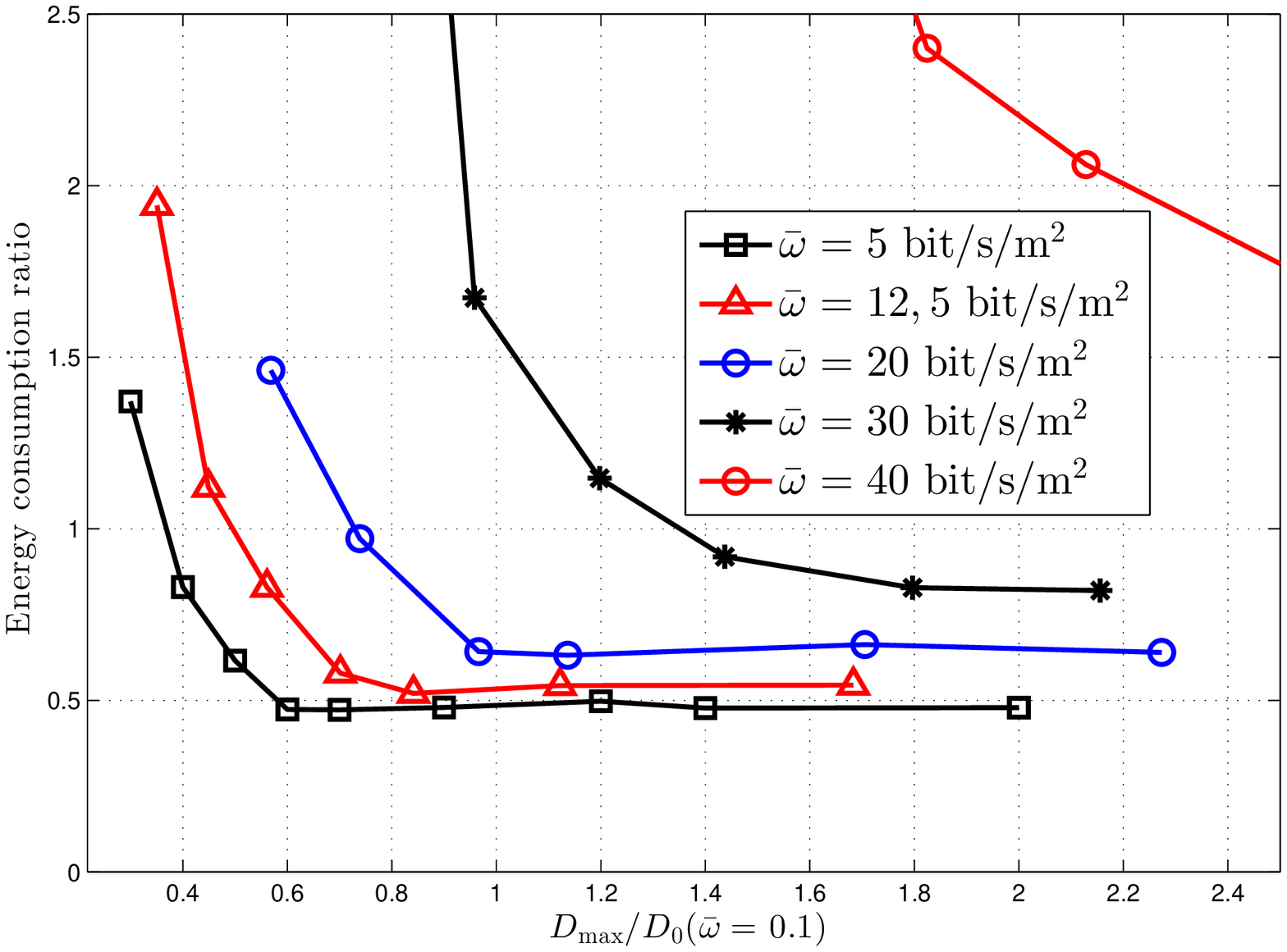}
\caption{Energy consumption ratio $\Pi (\boldsymbol{x}^*) /\Pi_0(\boldsymbol{x}_0^*(\bar{\omega}=5))$ vs delay $D_{\max}/ D_0(\bar{\omega}=5)$ for varying $\bar{\omega}$ ($n_{RN}=1$).}
\label{fig:load_final}
\end{minipage}
\end{figure}

\subsubsection{Energy Efficiency of RNs vs Small Cells}
Fig.~\ref{fig:back_effect} compares the results obtained with our system model, with those obtained using small cells instead of RNs (i.e., $\beta=0,\, D_{BL,j,k}=0,\, \forall j \in \mathcal{K}_{B},\, \forall k \in \mathcal{K}_{R}$), highlighting the difference of our work with respect to those dedicated to devices with wired (or ideal) backhaul. As expected, RNs allow for a smaller energy consumption gain compared to small cells, for a fixed $D_{\max}$. This is due to the transmission delay on the backhaul link, to the increased delay on the access link (due to $\beta<1$) and to the constraints on RN placement related to backhaul path-loss and shadowing. Performance of RNs is more impaired when the traffic density or $n_{RN}$ increase. 
However, RN deployment still yields consistent uplink energy consumption gains.     

\subsubsection{Effectiveness of Exterior Search with Penalty Function}
Fig.~\ref{fig:compar_ext} compares the results of the optimization using penalty function and exterior search, with those obtained by means of an interior search. Both the interior and the exterior search have been carried on with the same number of iterations, for an unbiased comparison. The exterior search proves to be more effective when the delay constraint is tight, as $\tilde{\Omega}$ is expected to be fragmented and $|\tilde{\Omega}| \ll |\Omega |$. No meaningful gain in terms of search effectiveness is observed when the constraint is loose, as in this case the two approaches tend to coincide.
\begin{figure}
\centering
\begin{minipage}{.47\textwidth}
\centering
\includegraphics[width=\linewidth]{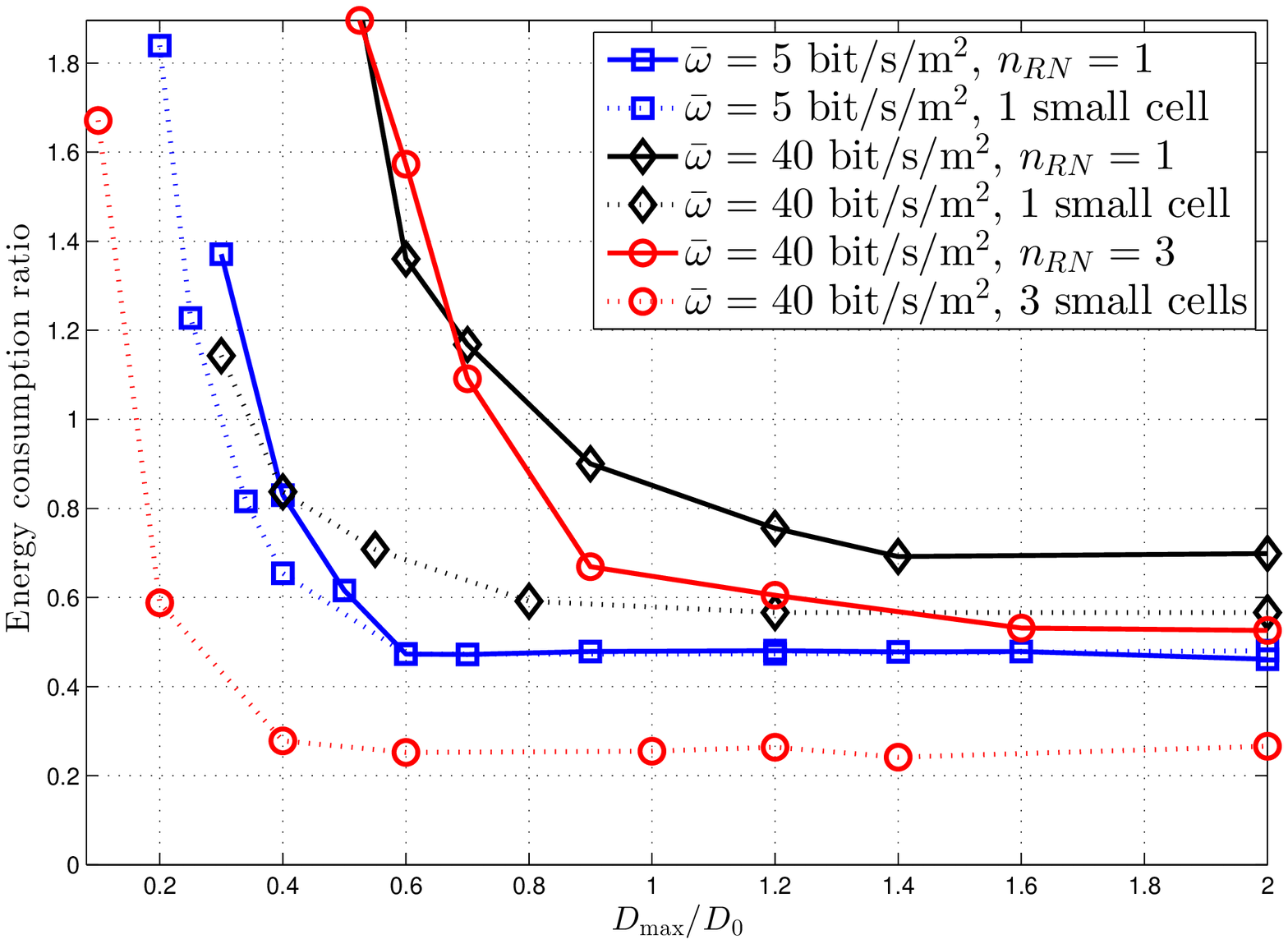} 
\caption{Comparison between wireless RNs and wired small cells, for varying $\bar{\omega}$ and $n_{RN}$ ($\beta=0.1$).}
\label{fig:back_effect}
\end{minipage}\hfill
\begin{minipage}{.47\textwidth}
\centering
\includegraphics[width=\linewidth]{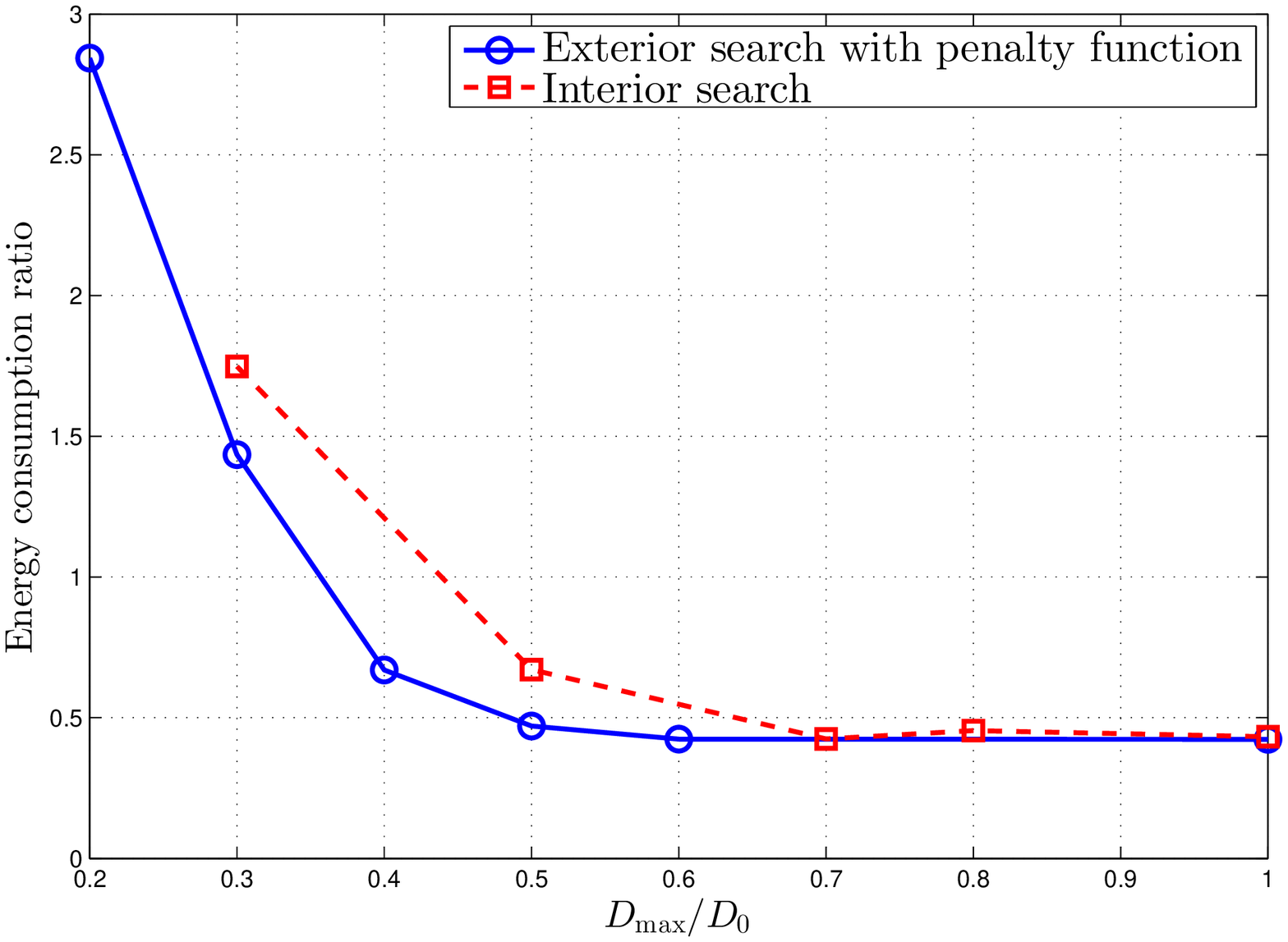}
\caption{Comparison between exterior and interior search approach effectiveness ($\bar{\omega}=5$~bits/s/m$^2$, $\beta=0.1$, $n_{RN}=1$).}
\label{fig:compar_ext}
\end{minipage}
\end{figure}

\subsubsection{Effect of Offered Traffic Spatial Distribution}\label{non_unif}
Fig.~\ref{fig:compar_gauss} compares the results obtained using a uniform $\phi(s)$ with those obtained when $\phi(s)$ is the sum of a uniform traffic profile and a bi-dimensional gaussian function, centered at $s:\{x=115 ; y=143 \}$~m, with the same average traffic density. We can see that RNs allow for larger energy consumption gains when traffic is non-uniform. This is due to the flexibility of the RN solution that allows relays to be placed close to the hot-spot center. 
\begin{figure}[]
 \centering
\includegraphics[width=\linewidth]{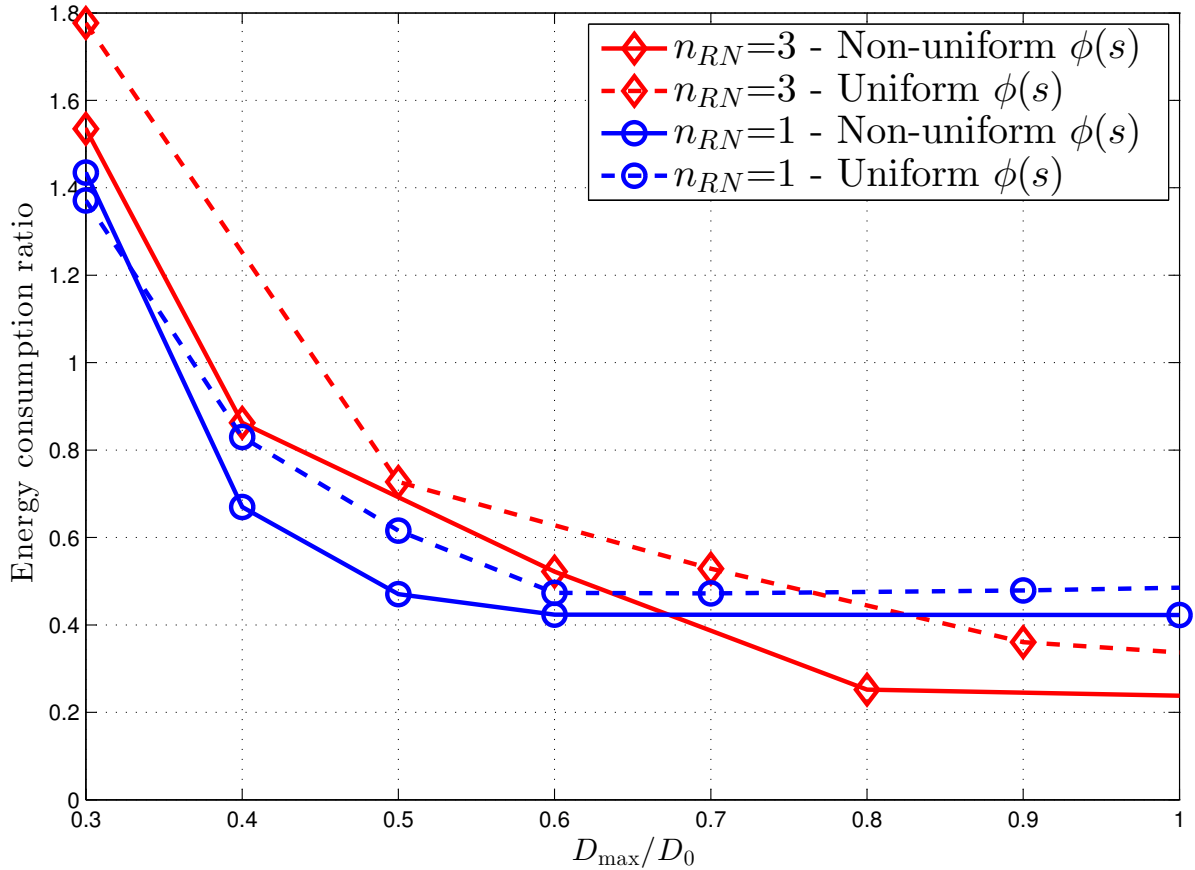} 
 \caption{Energy consumption vs delay tradeoff: $\bar{\omega}=0.1$, $\beta=0.1$, comparison of uniform and non-uniform traffic profile $\phi(s)$ for varying $n_{RN}$.}
\label{fig:compar_gauss} 
\end{figure} 


%% file: conclusions.tex
We have proposed a comprehensive framework for the optimization and performance evaluation of uplink energy consumption per bit in relays-enhanced cellular networks. This framework considers the arrival and departure of users, and the loads of network stations. Moreover, shadowing and fast fading are both taken into account in the propagation model. A realistic radio resource scheduling scheme is assumed, and its impact on users performance is derived. A customized optimization algorithm based on exterior search with penalty functions is proposed for the optimization, which is carried on under quality of service constraints. Results show a meaningful boosting of users energy efficiency given by the deployment of RNs, even if the average flow delay is imposed to be the same as in the network without relays. Proposed exterior search approach is shown to be more effective than the traditional interior search. 

%% file: appendix.tex
\section{Derivation of $\mu_{\gamma,k}$ and $\sigma_{\gamma,k}$}\label{appendix:momentmatch}

We use here a simple and classical approach consisting in introducing a auxiliary Random Variable (RV). (\ref{eq:sinrdef}) can indeed be rewritten as: $\gamma_k(s)=\frac{P_{k}(s)x}{\sum_{j \in \mathcal{K}, j\neq k} \Theta_j P_{k}(f_j)x +Nx}=P_k(s)xR_k$, where $x$ is a realization of a lognormal RV $X$, and $R_k^{-1}$ is the denominator. Note that $R_k$ does not depend on $s$. Several authors have shown that the interference term $I_k=\sum_{j\neq k} \Theta_j  P_{k}(f_j)$ can be well approximated by a lognormal distribution $\text{Log}{-}\mathcal{N}(\mu_{I_k},\sigma_{I_k}^2)$ even when the number of interferers is variable (see e.g. \cite{kostic, fischione}). In \cite{kostic}, it is also shown that the product of an exponential RV ($P_k(s)$) and a lognormal RV ($X$) as well as the sum of two lognormal RVs can be again approximated by a lognormal RV. Parameters $(\mu_{\gamma,k},\sigma_{\gamma,k})$ are then obtained by moment matching. The derivation of $(\mu_{I_k},\sigma_{I_k})$, which is less standard, is now detailed.
The mean $M_{1,I_k}$ of $I_k=\sum_{j\neq k} \Theta_j P_{k}(f_j)$ is equal to:
\begin{eqnarray} 
M_{1,I_k}&=&\mathbb{E}_{\tau}\left[ \sum_{j\neq k} \Theta_j P_{k}(f_j) \right] \notag \\
&=&\sum_{j\neq k} \rho_j \mathbb{E}_{\tau}\left[ P_k(f_j)\right] \notag \\
&\overset{(1)}{=}&\sum_{j\neq k} \rho_j \int_{\mathcal{A}_j} \frac{\varrho_j(s)}{\rho_j}T(s)G_k(s){ds} \notag \\
&=& \sum_{j\neq k} \rho_j Y_{j,k}, \notag
\end{eqnarray}
where $Y_{j,k}=\int_{\mathcal{A}_j} \frac{\varrho_j(s)}{\rho_j} T(s)G_{k}(s){ds}$. (1) is obtained by weighting the power received from location $s$ with the local load in $s$. Note that $T(s)G_{k}(s)$ are input parameters of the considered deployment scenario (propagation and transmit power assumptions). Now $\varrho_j(s)$ depends on the SINR distribution in $s$ and hence on the first moment $M_{1,I_k}$ of $I_k$. In order to avoid an additional complexity to our model, we make the approximation $\varrho_j(s)/\rho_j \approx \phi(s)/\Phi_j$. This is justified by the fact that in a urban environment, it is unlikely that UEs transmit at their maximum power \cite{bulakci13} so that every user of  $j$ is received with the same average power. The approximation comes then from the expressions (\ref{eq:rhoks}), (\ref{eq:rho_logn}) of $\varrho_j(s)$ and $\rho_j$. In the same way, we derive the variance $M_{2,I_k}$ of $I_k$: $M_{2,I_k} 
= \sum_{j\neq k} 2\rho_j H_{j,k} -\rho_j^2 Y_{j,k}^2$,
where $H_{j,k}=\int_{\mathcal{A}_j} \frac{\varrho_j(s)}{\rho_j} T^2(s)G_{k}^{2}(s){ds}$. Again, we approximate $\varrho(s)/\rho_j \approx \phi(s)/\Phi_j$. Finally, $\mu_{I_k}$ and $\sigma_{I_k}$ are found by matching $M_{1,I_k}$ and $M_{2,I_k}$ with the mean and variance of the approximating lognormal.

\section{Proof of Theorem~\ref{th:pi}: Derivation of $\pi_{\gamma,k}(s,z)$ under MQS scheduling} \label{appendix:MQS}
The MQS scheduler orders the values of SINR $\gamma_k(s,\tau-j),j \in \{0\cdots W-1\}$ of each UE in ascending order. 
The {\it ranking} $r_i(\tau)$ of UE $i$ located in $s_i$ on RB $\tau$ is the ranking of $\gamma_k(s_i,\tau)$ in the ordered vector $\gamma_k(s_i,\tau-j),j\in \{ 0\cdots W-1 \}$. The lower the ranking, the better the SINR on $\tau$ (wrt the SINR on previous blocks). 
Station $k$ assigns RB $\tau$ to the UE with the lowest $r_i(\tau), i \in \{1 \cdots U\}$. We assume $W$ large enough so that the probability for two or more UEs to have the same ranking is negligible.  
We have: $\mathbb{{P}}(r_{i}(\tau){=}n){=}1/W, \forall n \in\{1\cdots W\}$\footnote{as in our block Rayleigh fading environment each realization of $\gamma_k(s_i,\tau)$ is independent of the others.}. 
Moreover, the probability for any UE served by $k$ to be scheduled is $\mathbb{{P}}(f_k(\tau){=}s_i)=1/U,\forall i \in\{1 \cdots U\}$ (due to the fairness in RBs allocation). 
Let $\Pi_{\gamma,k}(s_i,z)=\int_{0}^{z}p_{\gamma,k}(s_i,t){dt}$ denote the Cumulative Distribution Function (CDF) of $\gamma_k(s_i,\tau)$. We have: 
\begin{eqnarray}
\lefteqn{\mathbb{{P}}\left(r_i(\tau){=}n \hspace{4pt} \middle | \hspace{4 pt} \gamma_{k}(s_i,\tau){=}z \right)=} \notag \\
&& \binom{W-1}{n-1}\Pi_{\gamma,k}(s_i,z)^{W-n}(1-\Pi_{\gamma,k}(s_i,z))^{n-1}. 
\end{eqnarray}
We denote $Q_n(i,z)$ the right hand side, which does not depend on $U$. 
Now, we derive $\pi_{\gamma,k}(s,z)$, {\it given} $U$. Applying the Bayes theorem on $\pi_{\gamma,k}(s,z)$ we obtain:
\begin{eqnarray}
    \pi_{\gamma,k}(s,z)\!\!\!\!&=& \!\!\!\! \frac{ \mathbb{{P}}\left( f_k(\tau){=}s_i \hspace{4pt} \middle | \hspace{4 pt} \gamma_k(s_i,\tau) {=}z \right)  p_{\gamma,k}(s_i,z)  }{\mathbb{{P}}(f_k(\tau){=}s_i)} \notag \\
     \!\!\!\!&=& \!\!\!\!\mathbb{{P}}\left( f_k(\tau){=}s_i \hspace{2pt} \middle | \hspace{2 pt} \gamma_k(s_i,\tau) {=}z \right) p_{\gamma,k}(s_i,z) U. \label{eq:bigeq2}
\end{eqnarray}
We then work on $\mathbb{{P}}\left( f_k(\tau){=}s_i \hspace{4pt} \middle | \hspace{4 pt} \gamma_k(s_i,\tau) {=}z \right)$ applying the law of total probability, conditioning it with respect to the possible rankings $r_i(\tau)$:
\begin{eqnarray}
  \lefteqn{\mathbb{{P}}\left( f_k(\tau){=}s_i \hspace{4pt} \middle | \hspace{4 pt} \gamma_k(s_i,\tau) {=}z \right)=} \notag \\
&&\displaystyle \sum_{n=1}^{W} \mathbb{{P}}\left( r_i(\tau){=}n \hspace{4pt} \middle | \hspace{4 pt} \gamma_{k}(s_i,\tau){=}z \right) \times \notag \\
&&\mathbb{{P}}\left(f_k(\tau){=}s_i \hspace{4pt} \middle | \hspace{4 pt} \gamma_{k}(s_i,\tau){=}z , r_i(\tau){=}n  \right).
\end{eqnarray}
The probability of being scheduled depends on the instantaneous SINR only through the ranking: once we know $r_i(\tau)$, knowing $\gamma_{k}(s_i,\tau)$ does not add any additional information regarding the probability of being scheduled by $k$.
Hence, 
\begin{eqnarray}
\lefteqn{\mathbb{{P}}\left(f_k(\tau){=}s_i \hspace{4pt} \middle | \hspace{4 pt} \gamma_{k}(s_i,\tau){=}z , r_i(\tau){=}n  \right)=} \notag \\
&& \mathbb{{P}}\left(f_k(\tau){=}s_i \hspace{4pt} \middle | \hspace{4 pt} r_i(\tau){=}n  \right)=\left(\frac{W-n}{W}\right)^{U-1},
\end{eqnarray}
and (\ref{eq:bigeq2}) becomes
\begin{equation}
    \pi_{\gamma,k}(s,z)= p_{\gamma,k}(s_i,z) U \displaystyle \sum_{n=1}^{W} Q_n(i,z) \left(\frac{W-n}{W}\right)^{U-1}.
\end{equation}

So far, we have found $\pi_{\gamma,k}(s,z)$ {\it given} $U$.
We now define $\Psi=(W-n)/W$, and use the law of total probability, summing up $\pi_{\gamma,k}(s,z)$ {\it given} $U$ for all the possible values of $U$, and obtaining its general expression:
\begin{eqnarray}
  \lefteqn{\mathbb{{P}}\left(\gamma_k(s_i,\tau) =z\hspace{4pt} \middle | \hspace{4 pt}  f_k(\tau){=}s_i, U>0 \right)=} \notag \\ 
  &&\!\!\!p_{\gamma,k}(s_i,z)  \displaystyle \sum_{n=1}^{W} Q_n(i,z) (1-\rho_k)\sum_{U=1}^{\infty} \rho_k^{U-1}  U\Psi^{U-1},
\end{eqnarray}
where
\begin{equation}
 (1-\rho_k)\sum_{U=1}^{\infty} \rho_k^{U-1}  U\Psi^{U-1}
=  \frac{W^2(1-\rho_k)}{\left(W-\rho_k(W-n)\right)^2}.
\end{equation}
We name $T_k(W,n)\triangleq W^2(1-\rho_k)/ \left(W-\rho_k(W-n)\right)^2$ and get the conclusion.

\section{Backhaul Rate Derivation}\label{appendix:back}

We start from the definition of $R_{BL,j,k}$ as 
$R_{BL,j,k}=1/\mathbb{E}_{\tau}[ \frac{1}{C( \gamma_{BL,j,k}(\tau))}]$,
where $\gamma_{BL,j,k}(\tau)$ is given by: 
\begin{equation}\label{eq:SINR_BL}
 \gamma_{BL,j,k}(\tau)=\frac{P_{BL,j}(s_{RN,j,k})}{  \sum_{h\in \mathcal{K}_B,h\neq j} \Theta_{BL,h}(\tau) P_{BL,j}(f_{BL,h}(\tau)) +N },
\end{equation}
where $\Theta_{BL,h}$ is a binary RV indicating whether the backhaul of $h$ is active on $\tau$, i.e.,  $P(\Theta_{BL,h}=1)= \rho_{BL,h}$; $P_{BL,j}(s)$ is the power received by eNB $j$ from a RN located in $s$; $s_{RN,j,k}$ indicates the location of RN $k$ served by eNB $j$ and $f_{BL,h}(\tau)$ yields the location of the RN scheduled by $h$ for backhaul transmission on $\tau$.

During each RB $\tau$ a RN $k$ in cell $j$ can be scheduled for uplink transmission with probability 
$ p_s(j,k)= \frac{\tilde{\rho}_{BL,j,k}}{\sum_t \tilde{\rho}_{BL,j,t}}$,
where $\tilde{\rho}_{BL,j,t}=(\bar{\omega}\Phi_t/(\beta R_{BL,j,t}))$.
Let $T=|\mathcal{K}_{B}|$ denote the total number of cells in our network, and $i_j(\tau), j\in \mathcal{K}_{B}, i_j(\tau)\in \{0 \cdots n_{RN}\}$ denote the index of the RN scheduled on the backhaul of eNB $j$ on RB $\tau$, where $i_j(\tau)=0$ means that no RN has been scheduled for transmission. 
Now, assuming that backhaul scheduling decisions in a given cell do not depend on those taken in other cells, we have that the probability of having a given set $\{i_1(\tau), \cdots i_T(\tau) \}$ of scheduled RNs on RB $\tau$ is equal to 
\begin{eqnarray}
\mathbb{P}(f_{BL,1}(\tau)=s_{RN,1,i_1}, \cdots , f_{BL,T}(\tau)=s_{RN,T,i_T})&=&\notag \\
\prod_{j} p_s(j,i_j(\tau)). &&
\end{eqnarray}
There are $(n_{RN}+1)^{T}$ possible scheduling combinations in all cells. We assign an index $\delta, \delta=\{1\cdots (n_{RN}+1)^{T}\}$ to each combination, and denote $\mathcal{V}(\delta)=\prod_{j} p_s(j,i_j(\delta))$. 
Finally, we apply the law of total probability to obtain
\begin{equation}
   R_{BL,j,k}=\displaystyle \frac{1}{\sum_{\delta=1}^{(n_{RN}+1)^{T}}\mathcal{V}(\delta) \hspace{2 pt}
                        \frac{1}{C\left(\frac{P_{BL,j}(s_{RN,j,k})}{  \sum_{h\in \mathcal{K}_B,h\neq j} P_{BL,j}(s_{RN,h,i_h(\delta)}) +N }\right)}},
\end{equation}
where we assume $P_{BL,j}(s_{RN,h,0})=0$.

\section{Some properties of Gibbs distributions for penalized energies} \label{app:gibbs}

The purpose of this Appendix is to give a hint to the "good" convergence 
of a SA process when an increasing exterior penalty 
such as (\ref{eq:adoptedpen}) is added in the global energy, 
coupled with an adequate SA temperature scheme. 
Consider such an augmented energy:
\begin{eqnarray}
\augmented(\config{x})&=&U(\config{x})  + \mu\ \potential(x), \\
\potential(\config{x})&=&\charact_{\constraint(\config{x}) > \constraintmax}\dfrac{\objective(\config{x})}{\constraintmax}(\constraint(\config{x}) - \constraintmax)
\end{eqnarray}
with for instance $\multiplier = \multiplier(m) = \alpha (m-1) \geq 0$.
We say that 
$\potential$ 
is a {\bf penalty} if the set of its minimizers
\begin{equation}
\Omega^* 
= \{\ x \in \Omega \st  \potential(\config{x}) 
= \potential^*= \min_{\config{y} \in \Omega} \potential(\config{y}) \ \}
\label{ms-minimizers:eq}
\end{equation}
is exactly the feasible subspace $\tilde{\Omega}$
(this holds for 
(\ref{eq:adoptedpen})). 
Now let us in general define the following set of "iso-constrained" subspaces:
$
\epsilon \in \Real 
\mapsto
\Omega_{\epsilon} = \{ \config{x} \in \Omega \st \potential(\config{x})  = \epsilon \}
,
$
and consider the exponential distribution
given in (\ref{eq:probaexp}) 
(see also \cite{geman-graffigne}).
For any value of
$\epsilon$ 
such that
$\Omega_\epsilon \neq \emptyset$,
one has:
\begin{eqnarray}
\forall \config{x},\ \config{y} \in \Omega_{\epsilon}, \dfrac{P(X = \config{y})}{P(X = \config{x})} &=&\dfrac{\exp -\augmented(\config{y})}{\exp -\augmented(\config{x})} \notag \\
&=&\dfrac {\exp - \left[\ U(\config{y}) +\multiplier\ \potential(\config{y})\ \right]}{\exp - \left[\ U(\config{x}) + \multiplier\ \potential(\config{x}) \ \right]} \notag \\
&=&\dfrac {\exp - U(\config{y})}{\exp - U(\config{x}) }.
\end{eqnarray}
So one can safely write in such an iso-constrained subspace:
\begin{equation}
P(X = \config{x} \cond \config{x} \in \Omega_{\epsilon}) = 
\dfrac {\exp - U(\config{x}) }{ \sum_{\config{y} \in \Omega_{\epsilon}} \exp - U(\config{y})}.
\label{ms-exponential:eq}
\end{equation}
which is also an exponential distribution. 

Two key points are now at stake in view of SA purposes 
\cite{geman-graffigne,robini-bresler-magnin}:
\begin{enumerate}
\item[$\bullet$]
First,
as $\multiplier \ra +\infty$
($\multiplier$ is similar to an inverse temperature associated to the constraint 
$\potential(\config{x}) $),
the global distribution 
(\ref{eq:probaexp}) 
becomes concentrated on the subspace 
$\Omega^*$ 
defined in (\ref{ms-minimizers:eq})
while keeping the exponential form (\ref{ms-exponential:eq}).
The proof is classical:
one can re-write
(\ref{eq:probaexp}) 
as
\begin{eqnarray}
P(X = \config{x}) \!\!\!&=&\!\!\! \dfrac {\exp - \left[\ U(\config{x}) +\multiplier\ \potential(\config{x})\ \right]}{\sum_{\config{y} \in \Omega} \exp - \left[\ U(\config{y}) + \multiplier\ \potential(\config{y}) \ \right]} \notag \\
\!\!\!&=&\!\!\! \dfrac{\exp - \left[\ U(\config{x})  + \multiplier\ ( \potential(\config{x}) -\potential^*)\ \right]}{
\sum_{\config{y} \in \Omega}{ \exp - \left[\ U(\config{y}) + \multiplier\ ( \potential(\config{y}) - \potential^*)\ \right]}} \notag
\end{eqnarray}
Now, either $\potential(\config{x}) > \potential^*$ and $\exp (-\ \multiplier\ ( \potential(\config{x})  - \potential^*))\mathop{\ra 0}$ when $\multiplier \ra +\infty$, or $\potential(\config{x})  = \potential^*$ and in this case, $\exp (-\ \multiplier\ ( \potential(\config{x})  - \potential^*))= 1\; \forall \multiplier$.
\item[$\bullet$]
Then, consider 
distribution (\ref{eq:probaexp}) 
with a temperature parameter $T > 0$.
It can be written as:
\begin{equation}
P_T(X = \config{x}) = 
\dfrac{\exp - [\ \dfrac{\ U(\config{x})  + \multiplier\ \potential(\config{x})}{T}\ ]}{
\sum_{\config{y} \in \Omega}{ 
\exp - [\ \dfrac{\ U(\config{y})  + \multiplier\ \potential(\config{y})}{T}\ ]}} \notag
\end{equation}
\begin{equation} 
=\dfrac{\exp - [\ \dfrac{(U(\config{x}) - \constraintMinU)}{T}  
               + \dfrac{\multiplier}{T}\ (\potential(\config{x}) - \potential^*)\ ]}
{\sum_{\config{y} \in \Omega} 
\exp - [\ \dfrac{(U(\config{y}) - \constraintMinU)}{T}  
               + \dfrac{\multiplier}{T}\ (\potential(\config{y}) - \potential^*)\ ]
}, \notag
\end{equation}
where 
$\constraintMinU$ 
is the minimum of value of objective function $U(.)$ on $\Omega^*$.
Now, if both 
$T \ra 0^+ \mbox{ and } \multiplier \ra +\infty 
\st
\multiplier' = \dfrac{\multiplier}{T} 
\sim \log m
\mbox{ (new inverse temperature)}
$,
a similar analysis as before,
now in two steps 
shows that for the penalty case
where  $\Omega^*= \feasable$,
the 
distribution of interest
$P_T(.)$
concentrates on those configurations with 
minimal energy
$\constraintMinU$
on feasible subspace $\tilde{\Omega}$ 
(see rigorous proof in 
\cite{geman-graffigne, robini-bresler-magnin}).
\end{enumerate}